\documentclass[preprint, review]{elsarticle}

\usepackage{amsmath,amssymb,amsthm,amsfonts}
\usepackage{mathtools}
\usepackage{bm}
\usepackage{graphicx}  
\usepackage{float}     
\usepackage{caption}
\usepackage{subcaption}
\usepackage{booktabs}   
\usepackage{multirow}
\usepackage{adjustbox}
\usepackage{array}
\usepackage{rotating}
\usepackage{orcidlink}
\usepackage{enumerate} 
\usepackage{siunitx}
\usepackage{url}
\usepackage{hyperref} 

\hypersetup{
    colorlinks=true,
    linkcolor=blue,
    citecolor=blue,
    urlcolor=blue,
    pdftitle={APHBXII},
    unicode=true
}

\newcommand{\cbar}{\overline{c}}
\newcommand{\Jbar}{\overline{J}}
\newcommand{\APHG}{\textsc{APH--G}}
\newcommand{\APHBXII}{\textsc{APHBXII}}

\newcommand{\BXII}{\textsc{BXII}}
\newcommand{\HBXII}{\textsc{HBXII}}
\newcommand{\APBXII}{\textsc{APBXII}}   
\newcommand{\MOBXII}{\textsc{MOBXII}}
\newcommand{\rth}{$r^{\text{th}}$}
\newcommand{\E}{\mathbb{E}}
\newcommand{\muprime}{\mu'}
\newcommand{\partialderiv}[2]{\frac{\partial #1}{\partial #2}}
\newtheorem{theorem}{Theorem}
\newtheorem{definition}{Definition}

\begin{document}

\begin{frontmatter}

\title{Alpha Power Harris--G Family of Distributions: Properties and Application to Burr XII Distribution}

\author[uniosun,fedpolyede]{Gbenga A. Olalude\corref{cor},\orcidlink{0000-0002-9950-0552}}
\ead{gbenga.olalude@federalpolyede.edu.ng}

\author[uniosun]{Taiwo A. Ojurongbe}

\author[uniosun]{Olalekan A. Bello}

\author[uniosun]{Kehinde A. Bashiru}

\author[uniosun]{Kazeem A. Alamu}

\cortext[cor]{Corresponding author}

\address[uniosun]{Department of Statistics, Osun State University, Oshogbo, Osun State, Nigeria}

\address[fedpolyede]{Department of Mathematics and Statistics, Federal Polytechnic Ede, Ede, Osun State, Nigeria}

\begin{abstract}
This study introduces a new family of probability distributions, termed the alpha power Harris-generalized (\APHG) family. The generator arises by incorporating two shape parameters from the Harris-G framework into the alpha-power transformation, resulting in a more flexible class for modelling survival and reliability data. A special member of this family, obtained using the two-parameter Burr XII distribution as the baseline, is developed and examined in detail. Several analytical properties of the proposed alpha power Harris Burr XII (\APHBXII) model are derived, which include closed-form expressions for its moments, mean and median deviations, Bonferroni and Lorenz curves, order statistics, and R\'enyi and Tsallis entropies. Parameter estimation is performed via maximum likelihood, and a Monte Carlo simulation study is carried out to assess the finite-sample performance of the estimators. In addition, three real lifetime datasets are analyzed to evaluate the empirical performance of the \APHBXII\ distribution relative to four competing models. The results show that the five-parameter \APHBXII\ model provides superior fit across all datasets, as supported by model-selection criteria and goodness-of-fit statistics.
\end{abstract}


\end{frontmatter}

\section{Introduction and motivation}
In recent decades, flexible distributions that capture the heterogeneity and complex behavior in lifetime and reliability data have been widely studied in fields like reliability engineering and survival modelling. This is because traditional models such as exponential, Rayleigh, and standard Weibull often fail to reflect the skewed and varying hazard rates seen in real datasets \cite{Warahena2023exponentiated, Shah2025new}. According to \cite{Lee2013methods}, popular methods of generating new families of distributions since the 1980s can be categorized as ‘methods of combination’ for the reason that these methods attempt to combine existing distributions to form new distributions or add new parameters to an existing distribution. 

The exponentiation method was among the earliest applications of the combination method, as categorized by \cite{Lee2013methods}. \cite{Mudholkar1993exponentiated} introduced the exponentiated--Weibull family by adding a second shape parameter to the Weibull distribution. The resulting distribution enabled the modelling of bathtub-shaped hazard rates that an ordinary Weibull could not accommodate. Another earlier approach was the parameter insertion method, as seen in the procedure of \cite{Marshall1997new}, who proposed adding a new parameter to an existing distribution to adjust its scale or shape. The result gave rise to the popular Marshall--Olkin family of distributions, which is a distribution that improves tail flexibility by compounding baseline lifetimes with an independent shock or frailty parameter \cite{Marshall1997new}. Another notable paradigm shift is the transformation of the cumulative distribution function (CDF) using bounded mixing distributions. \cite{Eugene2002beta} pioneered the Beta--generated (Beta--G) family by using the Beta distribution to transform the CDF of a baseline distribution. Their work introduced the addition of two shape parameters that govern the lower and upper tails of any newly generated distribution. Building on this, \cite{Cordeiro2011new} introduced the Kumaraswamy--G family, substituting the Beta distribution with a Kumaraswamy distribution to attain a comparable outcome while offering a more manageable cumulative distribution function (CDF).

In 2013, \cite{Alzaatreh2013new} presented a unifying framework known as the T--X family of distributions. Their proposed method provided a general recipe for constructing new families of continuous distributions by passing a baseline random variable X through a transformed cumulative distribution function (CDF), T. Meanwhile, many previous and recent continuous families of distributions can be derived as special cases of the T--X family \cite{Alzaatreh2013new}. These distributions, families of distributions, generators and related “--G” generators such as, the generalized odd Burr III--G by \cite{Haq2019generalized}, the odd Chen--G (OC--G) by \cite{Anzagra2020odd}, the inverse odd Weibull--G by \cite{Eghwerido2020inverse}, the sine extended odd Fréchet--G by \cite{Jamal2021sine}, the type II half logistic exponentiated--G by \cite{BELLO2021type}, the gamma odd Burr X--G (GOBX--G) by \cite{Tlhaloganyang2022gamma}, the compounded bell--G by \cite{Alsadat2023compounded}, the new generalized odd Fréchet--odd exponential--G by \cite{Sadiq2023new}, the Harris--odd Burr type X--G by \cite{Chipepa2023harris}, the gamma--Topp--Leone--type II--exponentiated half logistic--G by \cite{Oluyede2023gamma}, the Lomax--exponentiated odds ratio--G (L--EOR--G) by \cite{Roy2024lomax}, and the odd generalized Rayleigh reciprocal Weibull--G by \cite{Musekwa2025odd}, are popular, recent and novel generated families of distributions that have been used  to analyze lifetime data in several practical areas.

In 2017, \cite{Mahdavi2017new} introduced the alpha power transformation (APT) as a novel method for generating families of distributions. Their idea was inspired by a variation of the Poisson--G family proposed by \cite{Abouelmagd2017poisson}, where the parameter $\alpha = e^{\theta}$. The main purpose of the APT is to provide flexibility in the hazard function and incorporate skewness in a baseline distribution or generator through the addition of a shape parameter $\alpha$. The cumulative distribution function (CDF) of the APT is defined as

\begin{equation}
F_{\text{APT}}(x;\,\alpha) = \begin{cases} \displaystyle \frac{\alpha^{\,G(x)} - 1}{\alpha - 1}, & \alpha > 0,\; \alpha \ne 1, \\[10pt] G(x), & \alpha = 1, \end{cases} \qquad x \in \mathbb{R}.
\label{Eq:1}
\end{equation}

Empirical studies have shown that APT--extended distributions such as the APT--extended Burr II by \cite{Ogunde2020alpha}, the APT--log--logistic by \cite{Aldahlan2020alpha}, the APT--Dagum by \cite{Reyad2021alpha}, the APT--Lindley by \cite{Rosa2021alpha}, the APT--quasi Lindley by \cite{Udoudo2021alpha}, the APT--extended power Lindley by \cite{Eissa2022alpha}, the APT--generalized Weibull by \cite{Aga2024alpha}, the new APT--beta by \cite{Agegnehu2024new}, the generalized APT--truncated Lomax by \cite{El-Sayed2025generalized}, and the modified APT--inverse power Lomax by \cite{Akueson2025modified} provides excellent fits to lifetime data with minimal complexity penalty, and mostly outperforms other compared models in their works. Other researchers such as Ref. \cite{Afify2020heavytailed, Eghwerido2020alpha, Ali2021alphapower, Bulut2021alpha, Eghwerido2021alpha, Poonia2022alpha, Pimsap2023alpha, Qura2023novel, Ehiwario2023alpha} have also employed the transformation to propose the AP--exponentiated exponential, the AP--Gompertz, the AP–exponentiated inverse Rayleigh, the AP--Lomax, the AP--Teissier, the AP--exponentiated Teissier, the AP--exponentiated Pareto, the AP--Power Lomax, and the AP--Topp--Leone.

In other cases, the APT generator has been combined with other existing “--G” family to form new families of distributions. Using the dataset from \cite{Linhart1986statistik} on failure times of an airplane air-conditioning system, and the dataset from \cite{Aarset1987identify} on failure times of 50 devices, \cite{Elbatal2021alpha} examined the flexibility of the alpha power Transformed Weibull--G model by applying it to the Exponential, Rayleigh and Lindley distributions. Other examples are the AP--Marshall--Olkin--G \cite{Eghwerido2021alpha_b}, the extended AP--transformed--G \cite{Ahmad2021extended}, the AP--odd generalized exponential--G \cite{ELBATAL2022alpha}, the AP--Poisson--G \cite{Gbenga2022alpha}, the AP–Rayleigh–G \cite{Agu2022alpha}, the AP–Topp--Leone--G \cite{Eghwerido2022statistical}, the exponentiated generalized AP--G \cite{ElSherpieny2022exponentiated}, the modified AP--transformed Topp--Leone--G \cite{Ocloo2025modified}. Also, the AP--type II--G by \cite{Mohsin2025alpha}, the Novel AP--X by \cite{Bhat2025novel}, and the sine AP--G by \cite{Alghamdi2025sine} are examples of recent modifications of the APT generator.

The vast application of APT has proven it to be a powerful yet relatively parsimonious method for extending distributions. In this article, we propose a new family of distributions for generating lifetime models called the alpha power Harris--G (\APHG). The Harris family by \cite{Aly2011new} is categorized as a generalization of the Marshall--Olkin class when taking into account the probability generating function (PGF) of the Harris distribution \cite{Harris1948branching}. The Harris--G family stands out for its ability to unify skewness and tail-weight adjustments through its CDF;
\begin{equation}
G(x;c,\upsilon,\Pi) = 1 - \frac{c^{1/\upsilon}\,\overline{J}(x;\Pi)}{\left( 1 - \overline{c}\,\overline{J}(x;\Pi)^{\upsilon} \right)^{1/\upsilon}}, \qquad x>0,\; c>0,\; \upsilon>0,
\label{Eq:2}
\end{equation}

where $\overline{c}=1-c$, $\overline{J}(x;\Pi)=1-J(x;\Pi)$, and $\Pi$ denotes the parameter vector of the baseline distribution $J(x;\Pi)$.

Therefore, the newly proposed family of continuous distributions unifies the Harris
process with the alpha power transformation, enabling the capture of diverse data
patterns with fewer parameters while retaining analytical tractability and improving
the modelling of lifetime, reliability, and survival data.

The CDF of the proposed APH--G family is given by:
\begin{equation}
F(x;K) = \begin{cases} \displaystyle \frac{\alpha^{\,1 - \frac{c^{1/\upsilon}\,\overline{J}(x;\Pi)} {\left(1-\overline{c}\,\overline{J}(x;\Pi)^{\upsilon}\right)^{1/\upsilon}}} - 1}{\alpha - 1}, & \alpha>0,\; \alpha \ne 1, \\[12pt] \displaystyle 1 - \frac{c^{1/\upsilon}\,\overline{J}(x;\Pi)}{ \left(1-\overline{c}\,\overline{J}(x;\Pi)^{\upsilon}\right)^{1/\upsilon} }, & \alpha = 1, \end{cases} 
\label{Eq:3}
\end{equation}
and its corresponding pdf is given as
\begin{equation}
f(x;K) = \begin{cases} \displaystyle \frac{\log(\alpha)}{\alpha-1}\; \frac{c^{1/\upsilon}\, j(x;\Pi)}{\left(1-\overline{c}\,\overline{J}(x;\Pi)^{\upsilon}\right)^{1+1/\upsilon}} \;\alpha^{\,1 -\frac{c^{1/\upsilon}\,\overline{J}(x;\Pi)}{\left(1-\overline{c}\,\overline{J}(x;\Pi)^{\upsilon}\right)^{1/\upsilon}}}, & \alpha>0,\; \alpha\ne 1, \\[14pt] \displaystyle \frac{c^{1/\upsilon}\, j(x;\Pi)}{\left(1-\overline{c}\,\overline{J}(x;\Pi)^{\upsilon}\right)^{1+1/\upsilon}}, & \alpha = 1, \end{cases} \label{Eq:4}
\end{equation}
where $K = \alpha, c, \upsilon, \Pi$ ; $x \in \mathbb{R}$, $\alpha, c, \upsilon > 0$, and $\Pi$ parameterizes the baseline distribution in Eqs. (\ref{Eq:3}) and (\ref{Eq:4}). 

The survival function (SF), hazard rate function (HRF), reversed hazard rate function (RHRF) and cumulative hazard rate function (CHRF) of \APHG\ family are given respectively as,

\begin{equation}
S(x;K) = \frac{\alpha - \alpha^{\,1 - \frac{c^{1/\upsilon} \Jbar(x;\Pi)}{ (1-\cbar \Jbar(x;\Pi)^{\upsilon})^{1/\upsilon} }}}{\alpha-1}
\label{Eq:5}
\end{equation}

\begin{equation}
h(x;K) = \frac{\log(\alpha)c^{1/\upsilon} j(x;\Pi) \alpha^{\,1 - \frac{c^{1/\upsilon} \Jbar(x;\Pi)}{ (1-\cbar \Jbar(x;\Pi)^{\upsilon})^{1/\upsilon} }}}{ (1-\cbar \Jbar(x;\Pi)^{\upsilon})^{1+1/\upsilon} } \frac{1}{(\alpha - \alpha^{\,1 - \frac{c^{1/\upsilon} \Jbar(x;\Pi)}{ (1-\cbar \Jbar(x;\Pi)^{\upsilon})^{1/\upsilon} }})}
\label{Eq:6}
\end{equation}

\begin{equation}
r(x;K) = \frac{\log(\alpha) \alpha^{\,1 - \frac{c^{1/\upsilon} \Jbar(x;\Pi)}{ (1-\cbar \Jbar(x;\Pi)^{\upsilon})^{1/\upsilon} }}}{ \alpha^{\,1 - \frac{c^{1/\upsilon} \Jbar(x;\Pi)}{ (1-\cbar \Jbar(x;\Pi)^{\upsilon})^{1/\upsilon} }} - 1} \frac{c^{1/\upsilon} j(x;\Pi)}{(1-\cbar \Jbar(x;\Pi)^{\upsilon})^{1+1/\upsilon}}
\label{Eq:7}
\end{equation}

\begin{equation}
H(x;K) = -\log S(x;K) = \log\left(\frac{\alpha-1}{\alpha - \alpha^{\,1 - \frac{c^{1/\upsilon} \Jbar(x;\Pi)}{ (1-\cbar \Jbar(x;\Pi)^{\upsilon})^{1/\upsilon} }}}\right)
\label{Eq:8}
\end{equation}
Where $K = \alpha, c, \upsilon, \Pi$ ; $x \in \mathbb{R}$, $\alpha, c, \upsilon > 0$, $\alpha \ne 1$ and $\Pi$ parameterizes the baseline distribution in Eqs. (\ref{Eq:5}), (\ref{Eq:6}), (\ref{Eq:7}), and (\ref{Eq:8}).

Another unique objective of this study, is the introduction of a new lifetime model named as alpha power Harris Burr XII (\APHBXII) distribution. The proposed distribution is generated by adding the Burr XII distribution as a baseline in the newly generated \APHG\ family of distributions. This proposed distribution is much flexible and able to model real phenomena with decreasing, unimodal or bathtub hazard rate function (HRF). Several statistical features of the suggested distribution were derived and computed. Monte Carlo simulation study was conducted and the MLE method was employed to assess the behavior of the parameters of the newly derived model. Three lifetime datasets were examined to demonstrate the flexibility and adaptability of the proposed distribution over four other competing distributions.

The remainder of the article is structured as follows. Section \ref{Sec:2} presents the formulation and definition of some properties of the \APHBXII\ distribution, as well as density and the hazard function plots. Section \ref{Sec:3} further illustrate some other important statistical and mathematical features of the \APHBXII\ model. In Section \ref{Sec:4}, the MLE method was used to estimate the numerals of the unknown parameters of the new model. In Section \ref{Sec:5}, we presented the Monte Carlo simulation results to understand the asymptotic behavior of the maximum likelihood estimates of \APHBXII\ distribution. In Section \ref{Sec:6}, we presented the analysis of the lifetime data sets, and stated our concluding remarks in Section \ref{Sec:7}.

\section{The Alpha Power Harris Burr XII (APHBXII) distribution}\label{Sec:2} 
In this section, we introduce the Burr XII distribution, which has its CDF given as:
\begin{equation}
J(x;\phi,\eta)=1-(1+x^{\phi})^{-\eta}, \qquad x>0
\label{Eq:9}
\end{equation}
and the corresponding PDF given by
\begin{equation}
j(x;\phi,\eta)=\phi\eta x^{(\phi-1)} (1+x^{\phi})^{-\eta-1}, \qquad x>0;\phi,\eta>0
\label{Eq:10}
\end{equation}
This serves as a baseline for Eq. (\ref{Eq:3}) to establish the newly derived \APHBXII\ distribution. 

A random variable $X$ is said to follow the \APHBXII\ model, if its CDF is derived by inserting Eq. (\ref{Eq:9}) in Eq. (\ref{Eq:3}).
Let $T(x)=\frac{c^{1/\upsilon} (1+x^\phi)^{-\eta}}{(1-\cbar(1+x^\phi)^{-\eta\upsilon})^{1/\upsilon}}$, then
\begin{equation}
F(x;\aleph)=\frac{\alpha^{1-T(x)} - 1}{\alpha-1}, \qquad x>0
\label{Eq:11}
\end{equation}
The PDF corresponding to \APHBXII\ model is given as
\begin{equation}
f(x;\aleph)=\left(\frac{\log\alpha}{\alpha-1}\right) \frac{c^{1/\upsilon} \phi\eta x^{(\phi-1)} (1+x^{\phi})^{-\eta-1}}{\left(1-\cbar(1+x^{\phi})^{-\eta\upsilon}\right)^{1+1/\upsilon}} \alpha^{1-T(x)}
\label{Eq:12}
\end{equation}

where $\aleph = \alpha, c, \upsilon,\phi, \eta$ ; $x \in \mathbb{R}$, $\alpha, c, \upsilon,\phi, \eta > 0$, $\alpha \ne 1$ and $\phi, \eta$ parameterizes the Burr XII distribution.

The survival function (SF), hazard rate function (HRF), reversed hazard rate function (RHRF), and cumulative hazard rate function (CHRF) of the \APHBXII\ distribution are given, respectively, as follows:

\begin{equation}
S(x;\aleph)=\frac{\alpha-\alpha^{1-T(x)}}{\alpha-1}, \qquad x>0
\label{Eq:13}
\end{equation}

\begin{equation}
h(x;\aleph)=\frac{\log\alpha c^{1/\upsilon} \phi\eta x^{(\phi-1)} (1+x^{\phi})^{-\eta-1}}{\left(\alpha-\alpha^{1-T(x)}\right) \left(1-\cbar(1+x^{\phi})^{-\eta\upsilon}\right)^{1+1/\upsilon}} \alpha^{1-T(x)}
\label{Eq:14}
\end{equation}

\begin{equation}
r(x;\aleph)=\frac{\log\alpha \alpha^{1-T(x)}}{\alpha^{1-T(x)} - 1} \frac{c^{1/\upsilon} \phi\eta x^{(\phi-1)} (1+x^{\phi})^{-\eta-1}}{\left(1-\cbar(1+x^{\phi})^{-\eta\upsilon}\right)^{1+1/\upsilon}}
\label{Eq:15}
\end{equation}

\begin{equation}
H(x;\aleph)=-\log S(x;\aleph)=\log\left(\frac{\alpha-1}{\alpha-\alpha^{1-T(x)}}\right)
\label{Eq:16}
\end{equation}

\hyperref[Fig:1]{Fig. 1} displays the density plots of the \APHBXII\ distribution with different sets of parameter values, while \hyperref[Fig:2]{Fig. 2} show the plots of different curves for hazard rate function of the model. \hyperref[Fig:1]{Fig. 1} revealed that the PDF of the \APHBXII\ model is unimodal, right skewed with heavy and light tail. On the other hand, \hyperref[Fig:2]{Fig. 2} show increasing, decreasing, unimodal and bathtub shapes which demonstrate the flexibility of \APHBXII\ distribution to model various kinds of dataset with complex structures.

\begin{figure}[ht!]
    \centering
    \includegraphics[width=\textwidth]{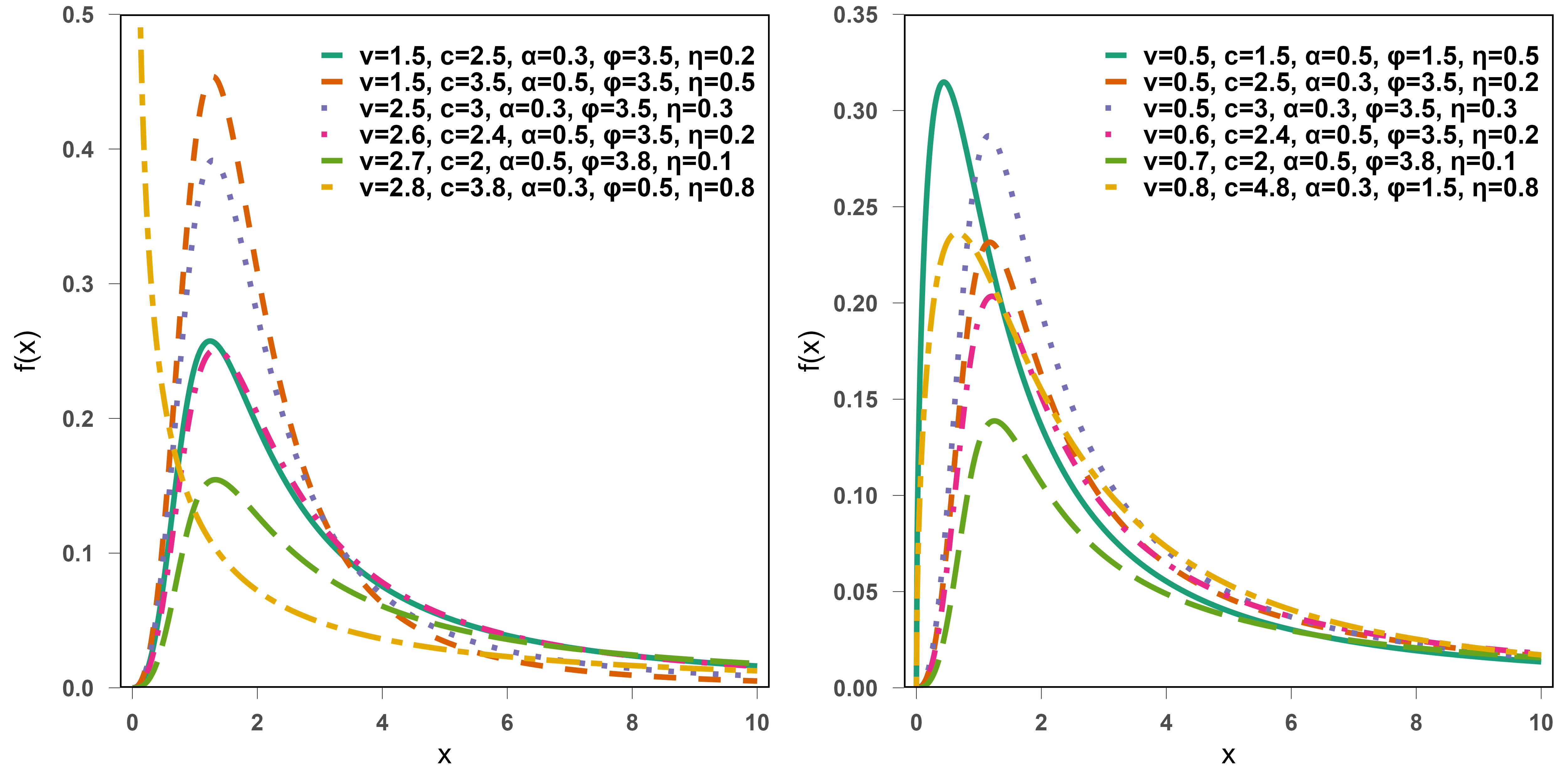} 
    \caption{Density curves of the APHBXII model. The plot displays the probability density function of the APHBXII model under different parameter settings, illustrating its flexibility in shape and tail behavior.}
    \label{Fig:1}
\end{figure}

\begin{figure}[ht!]
    \centering
    \includegraphics[width=\textwidth]{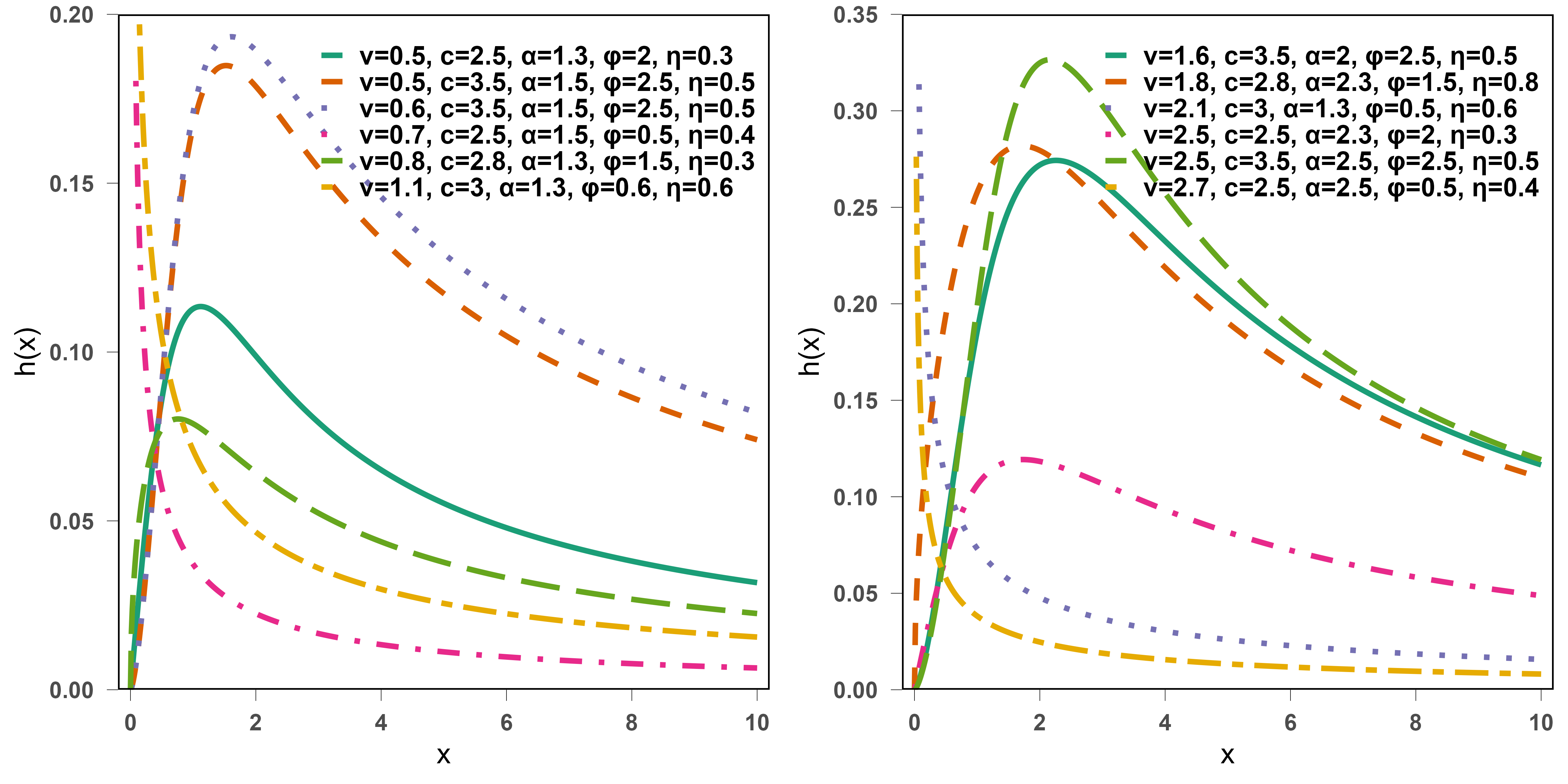}
    \caption{Hazard rate behavior of the APHBXII model. This figure shows examples of increasing, decreasing, and bathtub-shaped hazard functions produced by the APHBXII model.}
    \label{Fig:2}
\end{figure}

\section{Probabilistic features of APHBXII distribution}\label{Sec:3} 
In this section, we discuss some probabilistic features of the \APHBXII\ model.

\subsection{Quantile functions and applications}\label{subsec:3.1}

The quantile function for the \APHBXII\ distribution is defined as $Q(u;\aleph)=F^{-1} (u;\aleph)$, $u \in (0,1)$.

The $u^{\text{th}}$ quantile ($x_u$) of the \APHBXII\ model is derived by inverting the relation $F(x_u)=u$, hence from Eq. (\ref{Eq:11}), we have
\begin{equation}
u=\frac{\alpha^{1-T(x)} - 1}{\alpha-1}
\label{Eq:17}
\end{equation}

After performing algebraic calculations, we derived an expression for the quantile function of the \APHBXII\ model as follows:
\begin{equation}
x_u=\left[\left(\frac{1-\frac{\log[1+u(\alpha-1)]}{\log(\alpha)}}{(c+\cbar(1-\frac{\log[1+u(\alpha-1)]}{\log(\alpha)})^\upsilon)}\right)^{-1/(\eta\upsilon)} -1\right]^{1/\phi}
\label{Eq:18}
\end{equation}

Using Eq. (\ref{Eq:18}), we derive an expression for the lower quartile ($q_1$), median ($q_2$) and the upper quartile ($q_3$) of the \APHBXII\ distribution, and they are respectively given as
\begin{equation}
q_1=\left[\left(\frac{1-\frac{\log[1+0.25(\alpha-1)]}{\log(\alpha)}}{(c+\cbar(1-\frac{\log[1+0.25(\alpha-1)]}{\log(\alpha)})^\upsilon)}\right)^{-1/(\eta\upsilon)} -1\right]^{1/\phi}
\label{Eq:19}
\end{equation}

\begin{equation}
q_2=\left[\left(\frac{1-\frac{\log[1+0.5(\alpha-1)]}{\log(\alpha)}}{(c+\cbar(1-\frac{\log[1+0.5(\alpha-1)]}{\log(\alpha)})^\upsilon)}\right)^{-1/(\eta\upsilon)} -1\right]^{1/\phi}
\label{Eq:20}
\end{equation}

\begin{equation}
q_3=\left[\left(\frac{1-\frac{\log[1+0.75(\alpha-1)]}{\log(\alpha)}}{(c+\cbar(1-\frac{\log[1+0.75(\alpha-1)]}{\log(\alpha)})^\upsilon)}\right)^{-1/(\eta\upsilon)} -1\right]^{1/\phi}
\label{Eq:21}
\end{equation}

Eq. (\ref{Eq:18}) can be used to generate random numbers for the \APHBXII\ distribution. 

In \hyperref[Table:1]{Table 1}, we presented the values of the lower, median and upper quartiles of the \APHBXII\ model using some assumed values of the parameters, with $\alpha$ and $\upsilon$ fixed at 0.5 and 1.2 respectively. We also show the flexibility of skewness and kurtosis of \APHBXII\ distribution by determining the numerical values for Galton coefficient of skewness ($S$) \cite{galton1883enquiries} defined as $S=(Q(6/8)-2Q(4/8)+Q(2/8))/(Q(6/8)-Q(2/8) )$, which measures the degree of the long tail and the Moors coefficient of kurtosis ($K$) \cite{moors1988quantile} defined as $K=(Q(7/8)-Q(5/8)+Q(3/8)-Q(1/8))/(Q(6/8)-Q(2/8) )$, which determines the degree of tail heaviness of the \APHBXII\ distribution for the various assumed values of the parameters.

\begin{table*}[ht]
    \centering
     \caption{Values of $q_1, q_2, q_3, K, S$ of the APHBXII distribution.}
       \label{Table:1} 
    \begin{tabular}{@{}c *{5}{r}@{}}
    \toprule
    $c,\eta,\phi$ & $q_1$ & $q_2$ & $q_3$ & K & S
    \\
    \midrule
    0.3,0.3,1.2 & 0.3290 & 0.9550 & 3.6196 & 4.0068 & 0.6195 \\
    0.5,0.5,1.5 & 0.4029 & 0.9049 & 2.2821 & 2.2964 & 0.4657 \\
    0.8,0.8,2.0 & 0.4951 & 0.8729 & 1.5771 & 1.6170 & 0.3017 \\
    1.2,1.2,2.5 & 0.5580 & 0.8494 & 1.2738 & 1.3910 & 0.1858 \\
    2.0,2.0,3.5 & 0.6418 & 0.8373 & 1.0589 & 1.2880 & 0.0625 \\
    2.5,2.5,3.8 & 0.6555 & 0.8254 & 1.0052 & 1.2787 & 0.0281 \\
    4.0,4.0,5.0 & 0.7059 & 0.8219 & 0.9298 & 1.2888 & -0.0361 \\
    4.5,4.5,5.8 & 0.7356 & 0.8346 & 0.9237 & 1.2975 & -0.0529 \\
    5.5,5.5,7.0 & 0.7672 & 0.8462 & 0.9146 & 1.3117 & -0.0726 \\
    10.0,10.0,10.0 & 0.8095 & 0.8562 & 0.8944 & 1.3466 & -0.1018 \\
    \bottomrule
    \end{tabular}
\end{table*}

\subsection{Key mathematical expansions}\label{subsec:3.2}

In this subsection, we introduce some relevant mathematical expansions used in the derivation of some statistical properties of the \APHBXII\ model.

\subsubsection{Binomial expansion}\label{ssubsec:3.2.1}

The binomial expansion, for $i>0$ and $|m|<1$ is given by
\begin{equation}
(1-m)^i=\sum_{j=0}^i (-1)^j \binom{i}{j} m^j
\label{Eq:22}
\end{equation}
therefore, applying the binomial series expansion in Eq. (\ref{Eq:22}) to Eqs. (\ref{Eq:11}) and (\ref{Eq:12})respectively, we obtain the CDF of \APHBXII\ as,

\begin{equation}
F(x;\aleph)=\frac{1}{\alpha-1} \sum_{i=1}^\infty \frac{(\log\alpha)^i}{i!} \sum_{j=0}^i \binom{i}{j} (-1)^j c^{j/\upsilon} \left\{\frac{(1+x^{\phi})^{-\eta}}{(1-\cbar(1+x^{\phi})^{-\eta\upsilon})^{1/\upsilon}}\right\}^j
\label{Eq:23}
\end{equation}
and the corresponding PDF as,
\begin{equation}
f(x;\aleph)=\frac{\log\alpha}{\alpha-1} \sum_{i=0}^\infty \frac{(\log\alpha)^i}{i!} \sum_{j=0}^i \binom{i}{j} (-1)^j \frac{c^{((1+j)/\upsilon)} \phi\eta x^{(\phi-1)} (1+x^{\phi})^{-\eta-\eta j-1}}{\left(1-\cbar(1+x^{\phi})^{-\eta\upsilon}\right)^{1+1/\upsilon+j/\upsilon}}
\label{Eq:24}
\end{equation}

\subsubsection{Burr-XII kernel decomposition (triple series)}\label{ssubsec:3.2.2}
\begin{equation}
(1-\partial)^{-z}=\sum_{k=0}^\infty \binom{z+k-1}{k} \partial^k
\label{Eq:25}
\end{equation}
Applying the series given in Eqs. (\ref{Eq:25}) to (\ref{Eq:24}), we derive the mixture representation of the density function of \APHBXII\ distribution as
\begin{equation}
\label{Eq:26}
\begin{split}
f(x) &= \sum_{i=0}^\infty \sum_{j=0}^i \sum_{k=0}^\infty 
\underbrace{\left[\frac{\log\alpha}{\alpha-1} \frac{(\log\alpha)^i}{i!} \binom{i}{j} (-1)^j c^{\frac{1+j}{\upsilon}} (\bar{c})^k \binom{\frac{1+j}{\upsilon}+k}{k} \frac{\phi\eta}{\eta(j+\upsilon k+1)}\right]}_{\omega_{i,j,k}} \\
&\quad \times \underbrace{\eta(j+\upsilon k+1)x^{\phi-1} (1+x^\phi )^{-(\eta(j+\upsilon k+1)+1)}}_{g_{\phi,\eta(j+\upsilon k+1)} (x)}
\end{split}
\end{equation}
Summarily, Eq. (\ref{Eq:26}) becomes
\begin{equation}
f(x;\aleph)=\sum_{i=0}^\infty \sum_{j=0}^i \sum_{k=0}^\infty \omega_{i,j,k} g_{\phi,\eta(j+\upsilon k+1)} (x)
\label{Eq:27}
\end{equation}
where $g_{\phi,\eta(j+\upsilon k+1)} (x)$ represent the \BXII\ density with parameters $\phi$ and $\eta(j+\upsilon k+1)$.

Subsequently, by integrating Eq. (\ref{Eq:27}), we derive the mixture representation of the distribution function as,
\begin{equation}
F(x;\aleph)=\sum_{i=0}^\infty \sum_{j=0}^i \sum_{k=0}^\infty \omega_{i,j,k} G_{\phi,\eta(j+\upsilon k+1)} (x)
\label{Eq:28}
\end{equation}
where $G_{\phi,\eta(j+\upsilon k+1)} (x)$ is the CDF of the \BXII\ density with parameters $\phi$ and $\eta(j+\upsilon k+1)$.

Hence, the CDF and PDF of the mixture representation of \APHBXII\ model as expressed in  Eqs. (\ref{Eq:28}) and (\ref{Eq:28}) will be used to derived the expression of some basic properties of the model.

\subsection{Moments}\label{subsec:3.3}

In this subsection, ordinary and incomplete moments of the \APHBXII\ model were derived and discussed in detail.

\subsubsection{Ordinary moment of APHBXII distribution}\label{ssubsec:3.3.1}

\begin{theorem}
\begin{definition}If $X$ has the $\mathrm{APHBXII}(x;\aleph)$ $\forall$ $(i,j,k)$ $\omega_{i,j,k}\neq0$ : $r<\phi\eta (j+\upsilon k+1)$. A convenient sufficient condition (since $j,k\geq0$) is $0\leq r< \phi\eta$, then the raw $r$-th moment of APHBXII distribution is
\begin{equation*} 
\mu_r' = \sum_{i=0}^\infty \sum_{j=0}^i \sum_{k=0}^\infty \omega_{i,j,k} \frac{\eta(j+\upsilon k+1)}{\phi} B\left\{\left(\frac{r}{\phi}+1\right),\left(\eta(j+\upsilon k+1)-\frac{r}{\phi}\right)\right\}
\end{equation*}
\end{definition}
\end{theorem}

\begin{proof}
\begin{equation}
\mu_r' = \mathbb{E}[X^r]=\int_0^\infty x^r f(x;\aleph)dx
\label{Eq:29}
\end{equation}
inserting Eq. (\ref{Eq:27}) into (\ref{Eq:29}), we obtain
\begin{equation}
\mu_r' = \sum_{i=0}^\infty \sum_{j=0}^i \sum_{k=0}^\infty \omega_{i,j,k} \eta(j+\upsilon k+1) \int_0^\infty x^{(r+\phi-1)} (1+x^{\phi})^{-((\eta(j+\upsilon k+1)+1))} dx
\label{Eq:30}
\end{equation}
taking $y=x^{\phi}$, $x=y^{(1/\phi)}$, $dx=\frac{1}{\phi} y^{(1/\phi-1)} dy$, and substituting it in Eq. (\ref{Eq:30}), we get
\begin{equation}
\mu_r' = \sum_{i=0}^\infty \sum_{j=0}^i \sum_{k=0}^\infty \omega_{i,j,k} \frac{\eta(j+\upsilon k+1)}{\phi} \int_0^\infty y^{(r/\phi)} (1+y)^{-((\eta(j+\upsilon k+1)+1))} dy
\label{Eq:31}
\end{equation}
Furthemore, taking $z=y/(1+y)$, $y=z/(1-z)$, $dy=(1-z)^{-2} dz$ and inserting it in Eq. (\ref{Eq:31}), we obtain
\begin{equation}
\mu_r' = \sum_{i=0}^\infty \sum_{j=0}^i \sum_{k=0}^\infty \omega_{i,j,k} \frac{\eta(j+\upsilon k+1)}{\phi} \int_0^1 z^{(r/\phi)} (1-z)^{(\eta(j+\upsilon k+1)-r/\phi-1)} dz
\label{Eq:32}
\end{equation}
which after the beta transform, results in the expression for the $r$-th moment of the APHBXII distribution given by; 
\begin{equation}
\mu_r'=\sum_{i=0}^\infty \sum_{j=0}^i \sum_{k=0}^\infty \omega_{i,j,k} \frac{\eta(j+\upsilon k+1)}{\phi} B\left\{\left(\frac{r}{\phi}+1\right),\left(\eta(j+\upsilon k+1)-\frac{r}{\phi}\right)\right\}
\label{Eq:33}
\end{equation}
where $B(q,p)=\int_0^1 z^{(q-1)} (1-z)^{(p-1)} dz$ is a complete beta function $p>0,q>0$.
\end{proof}

The expression of the mean ($\mu_1'$) of $X$ is obtained by taking $r=1$, and it is given as
\begin{equation}
\mu_1'=\sum_{i=0}^\infty \sum_{j=0}^i \sum_{k=0}^\infty \omega_{i,j,k} \frac{\eta(j+\upsilon k+1)}{\phi} B\left\{\left(\frac{1}{\phi}+1\right),\left(\eta(j+\upsilon k+1)-\frac{1}{\phi}\right)\right\}
\label{Eq:34}
\end{equation}
The variance of $X$ can be further derived using the relation $\mu_2=\mu_2'-(\mu_1')^2$. Lastly, the \rth\ central moment and \rth\ cumulant of $X$ are defined respectively as,
\begin{align*}
\mu_r&=\E\{(X-\mu)^r \}=\sum_{q=0}^r \binom{r}{q} \mu_{r-q}' (-1)^q \mu^q \\
\kappa_r&=\mu_r'-\sum_{q=1}^{r-1} \binom{r-1}{q-1} \kappa_q \mu_{r-q}'
\end{align*}
with $\mu_1'=\kappa_1=\mu$. 

In \hyperref[Table:2]{Table 2}, we presented the results of the first six moments, variance ($\sigma^2$), the coefficient of variation ($\mathrm{CV}=\sigma/\mu=((\mu_2')/\mu^2 -1)^{(1/2)}$), coefficient of skewness ($S_{sk}$) and the coefficient of kurtosis ($S_{ku}$) computed using a fixed value of parameters $\phi=2.5$, $\eta=3.0$ and varying values of $\alpha,\upsilon$, and $c$. It can be observed from \hyperref[Table:2]{Table 2} that as the value of parameters $\phi$ and $\eta$ remains constant, the values of the coefficient of variation decreases as the values of $\alpha$, $c$ and $\upsilon$ increases.

\begin{table}[ht]
    \centering
     \small
    \caption{First six moments, $\sigma^2$, CV, $(S_{sk})$, $(S_{ku})$ of APHBXII distribution.}
    \label{Table:2}
    \begin{tabular}{@{}l cccc@{}}
    \toprule
    Moments & $\alpha=0.5, \upsilon=1.0$ & $\alpha=1.5, \upsilon=1.5$ & $\alpha=2.5, \upsilon=2.0$ & $\alpha=4.5, \upsilon=5.5$ \\
    \cmidrule(lr){2-2} \cmidrule(lr){3-3} \cmidrule(lr){4-4} \cmidrule(lr){5-5}
     & c=1.5 & c=2.5 & c=3.0 & c=3.5 \\
    \midrule
    $\mu_1'$ & 0.6447 & 0.8145 & 0.8626 & 0.8448 \\
    $\mu_2'$ & 0.5271 & 0.7940 & 0.8731 & 0.8353 \\
    $\mu_3'$ & 0.5305 & 0.9107 & 1.0273 & 0.9646 \\
    $\mu_4'$ & 0.6580 & 1.2413 & 1.4253 & 1.3218 \\
    $\mu_5'$ & 1.0391 & 2.0912 & 2.4310 & 2.2392 \\
    $\mu_6'$ & 2.3079 & 4.8374 & 5.6711 & 5.2087 \\
    $\sigma$ & 0.3339 & 0.3614 & 0.3591 & 0.3488 \\
    CV & 0.5180 & 0.4437 & 0.4163 & 0.4129 \\
    $S_{kw}$ & 1.2600 & 1.0856 & 1.1162 & 1.2616 \\
    $S_{ku}$ & 6.9392 & 6.7090 & 7.0704 & 7.4966 \\
    \bottomrule
    \end{tabular}
\end{table}

\subsubsection{Incomplete moment}\label{ssubsec:3.3.2}

\begin{theorem}
\begin{definition}
If $X$ follows the \APHBXII\ distribution and for any $t>0$ we define $X_t=X$ if $X\leq t$, and $X_t=0$ if otherwise, then the \rth\ incomplete moment of $X$ is equivalently expressed as
\begin{equation*}
\omega_r (t)= \sum_{i=0}^\infty \sum_{j=0}^i \sum_{k=0}^\infty \omega_{i,j,k} \frac{\eta(j+\upsilon k+1)}{\phi} \int_0^{z_t} z^{(r/\phi)} (1-z)^{(\eta(j+\upsilon k+1)-r/\phi-1)} dz
\end{equation*}
\end{definition}
\end{theorem}

\begin{proof}
\begin{equation}
\omega_r (t)=\E[(X_t )^r ]=\int_0^t x^r f(x;\aleph) dx
\label{Eq:35}
\end{equation}
inserting Eq. (\ref{Eq:27}) into (\ref{Eq:35}), we obtain
\begin{equation}
\omega_r (t)= \sum_{i=0}^\infty \sum_{j=0}^i \sum_{k=0}^\infty \omega_{i,j,k} \eta(j+\upsilon k+1) \int_0^t x^{(r+\phi-1)} (1+x^{\phi})^{-((\eta(j+\upsilon k+1)+1))} dx
\label{Eq:36}
\end{equation}
taking $y=x^{\phi}$, $x=y^{(1/\phi)}$, $dx=\frac{1}{\phi} y^{(1/\phi-1)}$ and substituting into Eq. (\ref{Eq:36}), we obtain 
\begin{equation}
\omega_r (t)= \sum_{i=0}^\infty \sum_{j=0}^i \sum_{k=0}^\infty \omega_{i,j,k} \frac{\eta(j+\upsilon k+1)}{\phi} \int_0^{(t^{\phi})} y^{(r/\phi)} (1+y)^{-((\eta(j+\upsilon k+1)+1))} dy
\label{Eq:37}
\end{equation}
Furthermore, taking $z=y/(1+y)$, $y=z/(1-z)$, $dy=(1-z)^{-2} dz$ and inserting it in Eq. (\ref{Eq:37}), we obtain
\begin{equation}
\omega_r (t)= \sum_{i=0}^\infty \sum_{j=0}^i \sum_{k=0}^\infty \omega_{i,j,k} \frac{\eta(j+\upsilon k+1)}{\phi} \int_0^{z_t} z^{(r/\phi)} (1-z)^{(\eta(j+\upsilon k+1)-r/\phi-1)} dz
\label{Eq:38}
\end{equation}
where $z_t=t^{\phi}/(1+t^{\phi})$ and $B_{z_t} (p,q) = \int_0^{z_t} u^{(p-1)} (1-u)^{(q-1)} du$ is an incomplete beta function with $p>0$, and $q>0$. This is valid $\forall$ $(i,j,k)$ where $\omega_{i,j,k}\neq0$ and $\eta (j+\upsilon k+1)-r/\phi >0$. A convenient sufficient condition (since $j,k\geq0$) is $0\leq r< \phi\eta$.

Therefore, the expression for the \rth\ incomplete moment of the \APHBXII\ distribution is obtained to be
\begin{equation}
\omega_r (t)= \sum_{i=0}^\infty \sum_{j=0}^i \sum_{k=0}^\infty \omega_{i,j,k} \frac{\eta(j+\upsilon k+1)}{\phi} B_{z_t} \left\{\left(\frac{r}{\phi}+1\right),\left(\eta(j+\upsilon k+1)-\frac{r}{\phi}\right)\right\}
\label{Eq:39}
\end{equation}
\end{proof}

Furthermore, taking $r=1$, we derive an expression for the first incomplete moment of \APHBXII\ distribution as
\begin{equation}
\omega_1 (t)= \sum_{i=0}^\infty \sum_{j=0}^i \sum_{k=0}^\infty \omega_{i,j,k} \frac{\eta(j+\upsilon k+1)}{\phi} B_{z_t} \left\{\left(\frac{1}{\phi}+1\right),\left(\eta(j+\upsilon k+1)-\frac{1}{\phi}\right)\right\}
\label{Eq:40}
\end{equation}

\subsection{Inequality and information measures}\label{subsec:3.4}

In this subsection, the Lorenz and Bonferroni curves, R\'enyi and Tsallis entropies of the \APHBXII\ model were introduced.

\subsubsection{Lorenz and Bonferroni curves for the APHBXII model}\label{ssubsec:3.4.1}

The expression given in Eq. (\ref{Eq:40}) above is important in the derivation of the Lorenz and Bonferroni curves which usefulness has been proven in insurance, reliability analysis and actuarial sciences.

The Bonferroni curve is defined by
\begin{equation}
B(p)=\frac{1}{p\mu} \int_0^u xf(x;\aleph)dx
\label{Eq:41}
\end{equation}
while the Lorenz curve is defined as follows
\begin{equation}
L(p)=\frac{1}{\mu} \int_0^u xf(x;\aleph)dx
\label{Eq:42}
\end{equation}
where $\mu=\E(X)$, which equate the first moment ($r=1$) and $u=F^{-1} (p)$. 

In our case, we obtain
\begin{equation}
\begin{split}
B(p)&=\frac{1}{p\mu} \left[\sum_{i=0}^\infty \sum_{j=0}^i \sum_{k=0}^\infty \omega_{i,j,k} \frac{\eta(j+\upsilon k+1)}{\phi} \right. \\
&\left. \qquad \times B_{z_t} \left\{\left(\frac{1}{\phi}+1\right),\left(\eta(j+\upsilon k+1)-\frac{1}{\phi}\right)\right\}\right]
\end{split}
\label{Eq:43}
\end{equation}
and 
\begin{equation}
\begin{split}
L(p)&=\frac{1}{\mu} \left[\sum_{i=0}^\infty \sum_{j=0}^i \sum_{k=0}^\infty \omega_{i,j,k} \frac{\eta(j+\upsilon k+1)}{\phi} \right. \\
&\left. \qquad \times B_{z_t} \left\{\left(\frac{1}{\phi}+1\right),\left(\eta(j+\upsilon k+1)-\frac{1}{\phi}\right)\right\}\right]
\end{split}
\label{Eq:44}
\end{equation}

\subsubsection{R\'enyi entropy}\label{ssubsec:3.4.2}

\cite{Renyi1961measures} introduced a significant mathematical formula that we used to obtain the entropy of the \APHBXII\ distribution given by
\begin{equation}
I_R^{\rho}=\frac{1}{1-\rho} \log\left[\int_0^\infty f_{\mathrm{APHBXII}} (x;\aleph)^{\rho} dx\right], \qquad \rho>0,\rho\neq1
\label{Eq:45}
\end{equation}
Inserting Eq. (\ref{Eq:12}) into (\ref{Eq:45}), we have
\begin{equation}
\begin{aligned}
I_R^{\rho}=\frac{1}{1-\rho} \log\left[\int_0^\infty \left(\left(\frac{\log\alpha}{\alpha-1}\right) \frac{c^{1/\upsilon} (\phi\eta) x^{(\phi-1)} (1+x^{\phi})^{-\eta-1}}{\left(1-\cbar(1+x^{\phi})^{-\eta\upsilon}\right)^{1+1/\upsilon}} \alpha^{1-T(x)}\right)^{\rho} dx\right]
\label{Eq:46}
\end{aligned}
\end{equation}
where,
\begin{align*}
\alpha^{\rho(1-T(x))} &= \alpha^{\rho} \sum_{m=0}^\infty \frac{(-\rho\log\alpha)^m}{m!} \frac{c^{m/\upsilon} (1+x^\phi)^{-m\eta}}{(1-\cbar(1+x^\phi)^{-\eta\upsilon})^{m/\upsilon}} 
\end{align*}
and
\begin{align*}
(1-\cbar(1+x^\phi)^{-\eta\upsilon})^{-\left(\rho(1+1/\upsilon)+m/\upsilon\right)}&=\sum_{q=0}^\infty \binom{\rho(1+1/\upsilon)+m/\upsilon+q-1}{q} \\
&\quad \times \cbar^q (1+x^\phi)^{-\eta\upsilon q}
\end{align*}
then, the integrand becomes $x^{(\rho(\phi-1))} (1+x^{\phi})^{-[\rho(\eta+1)+\eta m+\eta\upsilon q] }$.
\begin{equation}
W^{**}=\int_0^\infty x^{(\rho(\phi-1))} (1+x^{\phi})^{-[\rho(\eta+1)+\eta m+\eta\upsilon q]} dx
\label{Eq:47}
\end{equation}
taking $y=x^{\phi}$, $x=y^{(1/\phi)}$, $dx=\frac{1}{\phi} y^{(1/\phi-1)} dy$, and substituting it in (\ref{Eq:47}), we obtain
\begin{equation}
W^{**}=\frac{1}{\phi} \int_0^\infty y^{((\rho(\phi-1)+1-\phi)/\phi)} (1+y)^{-[\rho(\eta+1)+\eta m+\eta\upsilon q]} dy
\label{Eq:48}
\end{equation}
Furthermore, taking $y=z/(1-z)$, $dy=(1-z)^{-2} dz$ and inserting it in (\ref{Eq:48}), we obtain
\begin{equation}
W^{\text{**}} = \frac{1}{\phi} \int_0^1 z^{\frac{\rho(\phi-1)+1-\phi}{\phi}} (1-z)^{[\rho(\eta+1)+\eta m+\eta\upsilon q]-\frac{\rho(\phi-1)+1-\phi}{\phi}-2} \,dz
\label{Eq:49}
\end{equation}
Subsequently, after beta transform, Eq. (\ref{Eq:49}) becomes 
\begin{equation}
W^{\text{**}} = \frac{1}{\phi} B\left\{\left(\frac{\rho(\phi-1)+1}{\phi}\right),\left(\rho(\eta+1)+\eta m+\eta\upsilon q - \frac{\rho(\phi-1)+1-2\phi}{\phi}\right)\right\}
\label{Eq:50}
\end{equation}
Therefore,
\begin{equation}
\begin{split}
V^{\text{**}} &= \int_0^\infty f(x;\aleph)^\rho \,dx \\
&= \left(\frac{\log\alpha}{\alpha-1}\right)^\rho c^{\rho/\upsilon} (\phi\eta)^\rho \alpha^\rho \\
&\quad\times \sum_{m=0}^\infty \sum_{q=0}^\infty \frac{(-\rho\log\alpha)^m}{m!} c^{m/\upsilon} \binom{\rho(1+1/\upsilon)+m/\upsilon+q-1}{q} \bar{c}^q W^{\text{**}}
\end{split}
\label{Eq:51}
\end{equation}
Inserting Eq. (\ref{Eq:51}) into (\ref{Eq:45}), we derived an expression for the R\'enyi entropy of the \APHBXII\ model as
\begin{equation}
\begin{split}
I_R^\rho &= \frac{1}{1-\rho} \left[\rho\log\alpha+\rho\log\left(\frac{\log\alpha}{\alpha-1}\right)\right] + \frac{1}{1-\rho} \left[\log\left(c^{\rho/\upsilon} (\phi\eta)^\rho \right)\right] \\
&\quad+ \frac{1}{1-\rho} \left[\log\left(\sum_{m=0}^\infty \sum_{q=0}^\infty \frac{(-\rho\log\alpha)^m}{m!} c^{m/\upsilon} \binom{\rho(1+1/\upsilon)+m/\upsilon+q-1}{q}\right)\right] \\
&\quad+ \frac{1}{1-\rho} \left[\log\left(\bar{c}^q W^{\text{**}} \right)\right]
\end{split}
\label{Eq:52}
\end{equation}

\subsubsection{Tsallis entropy}\label{ssubsec:3.4.3}

The Tsallis entropy was first presented by \cite{Havrada1967quantification} and later developed by \cite{Tsallis1988possible}. The Tsallis entropy of the \APHBXII\ model can be expressed as
\begin{equation}
I_T^{\rho}=\frac{1}{\rho-1} \left[1-\int_0^\infty f_{\mathrm{APHBXII}} (x;\aleph)^{\rho} dx\right], \qquad \rho>0,\rho\neq1
\label{Eq:53}
\end{equation}
It then follows from (\ref{Eq:51}) that Eq. (\ref{Eq:53}) becomes
\begin{equation}
I_T^{\rho}=\frac{1}{\rho-1} [1-V^{**} ]
\label{Eq:54}
\end{equation}

\subsection{Probability weighted moment (PWM) and other distributional properties}\label{subsec:3.5}

In this subsection, probability weighted moment (PWM), order statistics, moment generating function (MGF), mean deviation, mean residual life, average waiting time of the \APHBXII\ model were derived and discussed in detail.

\subsubsection{Probability weighted moment (PWM)}\label{ssubsec:3.5.1}

Probability-weighted moment (PWM) is defined as the expectation of a prescribed transformation of a stochastic variate $X$, with $F$ denoted as its cumulative distribution function. Since the \APHBXII\ model possesses a closed form quantile function, the corresponding PWM integral can be evaluated analytically, which allow the PWMs of the \APHBXII\ model to be expressed explicitly in terms of its distribution parameters, as shown below:

\begin{equation}
\zeta_{q,r}=\E(X^q F(X)^r )=\int_0^\infty x^q F(x)^r f(x)dx
\label{Eq:55}
\end{equation}
Inserting Eqs. (\ref{Eq:11}) and (\ref{Eq:12}) into (\ref{Eq:55}), we have
\begin{equation}
\begin{aligned}
\zeta_{q,r}&=\frac{c^{1/\upsilon} \phi\eta\log\alpha}{(\alpha-1)^{r+1}} \int_0^\infty \left(\alpha^{1-T(x)} -1\right)^r \frac{x^{(q+\phi-1)} (1+x^{\phi})^{-\eta-1}}{\left(1-\cbar(1+x^{\phi})^{-\eta\upsilon}\right)^{1+1/\upsilon}} \alpha^{1-T(x)} dx
\label{Eq:56}
\end{aligned}
\end{equation}
Summarizing Eq. (\ref{Eq:56}) using series expansion, we have
\begin{equation}
\begin{aligned}
\left(\alpha^{1-T(x)} -1\right)^r &= \sum_{m=0}^r \binom{r}{m} \alpha^m (-1)^{r-m} \sum_{i=0}^\infty \frac{(-m \log\alpha)^i}{i!} \\
&\quad\times \left\{\frac{c^{1/\upsilon} (1+x^{\phi})^{-\eta}}{(1-\cbar(1+x^{\phi})^{-\eta\upsilon})^{1/\upsilon}}\right\}^i
\end{aligned}
\label{Eq:57}
\end{equation}
where,
\begin{equation*}
d_i=\frac{1}{i!} \sum_{m=0}^r \binom{r}{m} \alpha^m (-1)^{r-m} (-m \log\alpha)^i
\end{equation*}
then, 
\begin{equation}
\left(\alpha^{1-T(x)} -1\right)^r=\sum_{i=0}^\infty d_i \left\{\frac{c^{1/\upsilon} (1+x^{\phi})^{-\eta}}{(1-\cbar(1+x^{\phi})^{-\eta\upsilon})^{1/\upsilon}}\right\}^i
\label{Eq:58}
\end{equation}
and
\begin{equation}
\alpha^{1-T(x)} =\alpha\sum_{s=0}^\infty \frac{(-\log\alpha)^s}{s!} \left\{\frac{c^{1/\upsilon} (1+x^{\phi})^{-\eta}}{(1-\cbar(1+x^{\phi})^{-\eta\upsilon})^{1/\upsilon}}\right\}^s
\label{Eq:59}
\end{equation}
Multiplying the series in Eqs. (\ref{Eq:58}) and (\ref{Eq:59}) through Cauchy convolution, we get
\begin{equation}
\left(\alpha^{1-T(x)} -1\right)^r \alpha^{1-T(x)} = \alpha\sum_{l=0}^\infty e_l \left\{\frac{c^{1/\upsilon} (1+x^{\phi})^{-\eta}}{(1-\cbar(1+x^{\phi})^{-\eta\upsilon})^{1/\upsilon}}\right\}^l
\label{Eq:60}
\end{equation}
where,
\begin{equation*}
e_l=\sum_{i=0}^l d_i \frac{(-\log\alpha)^{(l-i)}}{(l-i)!}
\end{equation*}
\begin{equation*}
e_l=\sum_{i=0}^l \frac{1}{i!(l-i)!} \sum_{m=0}^r \binom{r}{m} (-1)^{(r-m)} \alpha^m (-\log\alpha m)^i (-\log\alpha)^{(l-i)}
\end{equation*}
from Eqs. (\ref{Eq:56}) and Eq. (\ref{Eq:60}), using the generalized binomial expansion,we have
\begin{equation}
\left(1-\cbar(1+x^{\phi})^{-\eta\upsilon}\right)^{-[1+((1+l)/\upsilon)]}=\sum_{k=0}^\infty \binom{1+((1+l)/\upsilon)+k-1}{k} \cbar^k (1+x^{\phi})^{-\eta\upsilon k}
\label{Eq:61}
\end{equation}
Hence, Eq. (\ref{Eq:56}) becomes,
\begin{equation}
\begin{aligned}
\zeta_{q,r}&=\frac{\alpha c^{1/\upsilon} \phi\eta\log\alpha}{(\alpha-1)^{r+1}} \sum_{l=0}^\infty e_l c^{(l/\upsilon)} \sum_{k=0}^\infty \binom{1+((1+l)/\upsilon)+k-1}{k} (\cbar)^k \\
&\quad\times\int_0^\infty x^{(q+\phi-1)} (1+x^{\phi})^{-[\eta(l+\upsilon k+1)+1]} dx
\end{aligned}
\label{Eq:62}
\end{equation}
taking $y=x^{\phi}$, $x=y^{(1/\phi)}$, $dx=\frac{1}{\phi} y^{(1/\phi-1)} dy$, and subsequently $y=z/(1-z)$, $dy=(1-z)^{-2} dz$ and substituting it into the integrand in Eq. (\ref{Eq:62}), we obtain
\begin{equation}
\int_0^\infty x^{(q+\phi-1)} (1+x^{\phi})^{-[\eta(l+\upsilon k+1)+1]} dx = \frac{1}{\phi} B\left\{\frac{q}{\phi}+1,\eta(l+\upsilon k+1)-\frac{q}{\phi}\right\}
\label{Eq:63}
\end{equation}
Inserting Eq. (\ref{Eq:63}) in (\ref{Eq:62}), we have,
\begin{equation}
\begin{aligned}
\zeta_{q,r}&=\frac{\alpha c^{1/\upsilon} \eta\log\alpha}{(\alpha-1)^{r+1}} \sum_{l=0}^\infty e_l c^{l/\upsilon} \sum_{k=0}^\infty \binom{((1+l)/\upsilon)+k}{k} \\
&\quad\times \cbar^k B\left\{\frac{q}{\phi}+1,\eta(l+\upsilon k+1)-\frac{q}{\phi}\right\}
\end{aligned}
\label{Eq:64}
\end{equation}

\subsubsection{Order statistics}\label{ssubsec:3.5.2}

For $r = 1,\dots,n$, the minimum order statistic corresponds to $r = 1$, while the maximum order statistic is obtained by equating $r = n$. Suppose $X_{(1)}, X_{(2)},\dots \\ ,X_{(n)}$ are the ordered variates drawn from the \APHBXII\ model; then the density function of the $r$-th order statistics $X_{(r)}$ is estimated as
\begin{equation}
f_r (x;\zeta)=\frac{1}{B(r,n-r+1)} F (x;\aleph)^{(r-1)} [1-F (x;\aleph)]^{(n-r)} f (x;\aleph)
\label{Eq:65}
\end{equation}
applying the series expansion given in Eq. (\ref{Eq:22}) in Eq. (\ref{Eq:65}), we obtain 
\begin{equation}
f_r (x;\aleph)=\frac{f (x;\aleph)}{B(r,n-r+1)} \sum_{i=0}^{(n-r)} (-1)^i \binom{n-r}{i} F (x;\aleph)^{(r+i-1)}
\label{Eq:66}
\end{equation}
Subsequently, inserting Eqs. (\ref{Eq:11}) and (\ref{Eq:12}) in Eq. (\ref{Eq:66}), we have
\begin{equation}
\begin{aligned}
f_r (x;\aleph)&=\frac{c^{1/\upsilon} \phi\eta\log\alpha}{B(r,n-r+1)} \sum_{i=0}^{(n-r)} (-1)^i \binom{n-r}{i} \frac{(\alpha^{1-T(x)}-1)^{(r-1+i)}}{(\alpha-1)^{(r+i)}} \\
&\quad\times\frac{x^{(\phi-1)} (1+x^{\phi})^{-\eta-1}}{(1-\cbar(1+x^{\phi})^{-\eta\upsilon})^{1+1/\upsilon}} \alpha^{1-T(x)}
\end{aligned}
\label{Eq:67}
\end{equation}
Further simplification of Eq. (\ref{Eq:67}) using the series expansions stated in Eqs. (\ref{Eq:22}) and (\ref{Eq:25}), gives rise to an expression for the \rth\ order statistics of the \APHBXII\ model, given as
\begin{equation}
\begin{aligned}
(\alpha^{1-T(x)}-1)^{(r-1+i)}&=\sum_{m=0}^{(r-1+i)} \binom{r-1+i}{m} (-1)^{(r-1+i-m)} \alpha^m \\
&\quad\times\sum_{u=0}^\infty \frac{(-m \log\alpha)^u}{u!} \{T(x)\}^u
\end{aligned}
\label{Eq:68}
\end{equation}
where,
\begin{equation*}
A_{i,m,u}=\binom{r-1+i}{m} (-1)^{(r-1+i-m)} \alpha^m \frac{(-m \log\alpha)^u}{u!}
\end{equation*}
Using $A_{i,m,u}$ as defined above, Eq. (\ref{Eq:68}) becomes,
\begin{equation}
(\alpha^{1-T(x)}-1)^{(r-1+i)}=\sum_{m=0}^{(r-1+i)} \sum_{u=0}^\infty A_{i,m,u} \{T(x)\}^u
\label{Eq:69}
\end{equation}
From Eq. (\ref{Eq:59}),
\begin{equation*}
\alpha^{1-T(x)} = \alpha\sum_{s=0}^\infty \frac{(-\log\alpha)^s}{s!} \{T(x)\}^s
\end{equation*}
Therefore, multiplying the series in Eqs. (\ref{Eq:58}) and (\ref{Eq:59}) through Cauchy convolution, we get,
\begin{equation}
(\alpha^{1-T(x)}-1)^{(r-1+i)} \alpha^{1-T(x)} = \alpha\sum_{l=0}^\infty B_{i,l} \{T(x)\}^l
\label{Eq:70}
\end{equation}
where,
\begin{equation*}
B_{i,l}=\sum_{u=0}^l \sum_{m=0}^{(r-1+i)} A_{i,m,u} \frac{(-\log\alpha)^{(l-u)}}{(l-u)!}
\end{equation*}
$B_{i,l}$ can further be expanded to,
\begin{equation*}
B_{i,l}=\frac{(-1)^l (\log\alpha)^l}{l!} \sum_{m=0}^{(r-1+i)} \binom{r-1+i}{m} (-1)^{(r-1+i-m)} \alpha^m (1+m)^l
\end{equation*}
Also, from Eqs. (\ref{Eq:67}) and (\ref{Eq:70}), using the generalized binomial expansion,we obtain
\begin{equation}
\left(1-\cbar(1+x^{\phi})^{-\eta\upsilon}\right)^{-[1+((1+l)/\upsilon)]}=\sum_{k=0}^\infty \binom{1+((1+l)/\upsilon)+k-1}{k} \cbar^k (1+x^{\phi})^{-\eta\upsilon k}
\label{Eq:71}
\end{equation}

Therefore, the fully expanded hyper-series for $f_r (x;\aleph)$ in terms of Burr-XII type kernels is given as,
\begin{equation}
\begin{aligned}
f_r (x;\aleph)&=\frac{\phi\eta\log\alpha}{B(r,n-r+1)} \sum_{i=0}^{(n-r)} (-1)^i \binom{n-r}{i} (\alpha-1)^{-(r+i)} \alpha\sum_{l=0}^\infty B_{i,l} c^{((1+l)/\upsilon)} \\
&\quad\times\sum_{k=0}^\infty \binom{((1+l)/\upsilon)+k}{k} (\cbar)^k x^{(\phi-1)} (1+x^{\phi})^{-[\eta(l+\upsilon k+1)+1]}
\end{aligned}
\label{Eq:72}
\end{equation}

\subsubsection{Moment generating function (MGF) of APHBXII model}\label{ssubsec:3.5.3}

The mgf of $\mathrm{APHBXII}(\aleph)$, say $M_X (t)$ is derived using
\begin{equation}
M_X (t)=\E(e^{tX} )=\int_0^\infty e^{tX} f(x,\aleph)dx=\sum_{r=0}^\infty \frac{t^r}{r!} \muprime_r
\label{Eq:73}
\end{equation}
Substituting Eq. (\ref{Eq:33}) into (\ref{Eq:73}), we obtain
\begin{equation}
\begin{aligned}
M_X (t)&=\sum_{i=0}^\infty \sum_{j=0}^i \sum_{k=0}^\infty \omega_{i,j,k} \frac{\eta(j+\upsilon k+1)}{\phi} \\
&\quad\times\sum_{r=0}^\infty \frac{t^r}{r!} B\left\{\left(\frac{r}{\phi}+1\right),\left(\eta(j+\upsilon k+1)-\frac{r}{\phi}\right)\right\}
\end{aligned}
\label{Eq:74}
\end{equation}

Likewise, the characteristic function of the \APHBXII\ distribution can be derived based on the \rth\ moments of the \APHBXII\ model:
\begin{equation}
\varphi_X (t)=\E(e^{itx} )=\int_{-\infty}^\infty e^{itx} f(x)dx=\sum_{r=0}^\infty \frac{(it)^r}{r!} \E(X^r)
\label{Eq:75}
\end{equation}
Inserting Eq. (\ref{Eq:33}) into (\ref{Eq:75}), an expression for the characteristic function of the \APHBXII\ model is derived as
\begin{equation}
\begin{aligned}
\varphi_X (t)&=\sum_{i=0}^\infty \sum_{j=0}^i \sum_{k=0}^\infty \omega_{i,j,k} \frac{\eta(j+\upsilon k+1)}{\phi} \\
&\quad\times\sum_{r=0}^\infty \frac{(it)^r}{r!} B\left\{\left(\frac{r}{\phi}+1\right),\left(\eta(j+\upsilon k+1)-\frac{r}{\phi}\right)\right\}
\end{aligned}
\label{Eq:76}
\end{equation}

\subsubsection{Mean deviation of the APHBXII model}\label{subsubsec:3.5.4}

The degree of dispersion in a population is assessed using the mean deviation, which is similar to the median. Suppose $M$ represents the median, and $\mu = \bar{X}$ represents the mean of the \APHBXII\ model as stated in Eqs. (\ref{Eq:20}) and (\ref{Eq:34}). 

The mean deviation of \APHBXII\ model about the mean is obtained as follows:
\begin{equation}
W_1 (X)=\E|X-\mu|=\int_0^\infty |X-\mu|f(x;\aleph)dx,
\label{Eq:77}
\end{equation}
Simplifying Eq. (\ref{Eq:77}), we have
\begin{equation}
W_1 (X) = 2\mu F(\mu;\aleph)-2\mu+2\int_\mu^\infty xf(x;\aleph)dx
\label{Eq:78}
\end{equation}
\begin{equation}
\begin{aligned}
W_1 (X)&=2\mu \frac{\alpha^{1-T(\mu)} - 1}{\alpha-1} -2\mu+2\sum_{i=0}^\infty \sum_{j=0}^i \sum_{k=0}^\infty \omega_{i,j,k} \int_\mu^\infty x g_{\phi,\eta(j+\upsilon k+1)} (x)dx
\end{aligned}
\label{Eq:79}
\end{equation}
\begin{equation}
W_1 (X)=2\mu \frac{\alpha^{1-T(\mu)} - 1}{\alpha-1} -2\mu+2 (\muprime_1-\omega_1 (t))
\label{Eq:80}
\end{equation}
\begin{equation}
\begin{aligned}
W_1 (X)&=2\mu \frac{\alpha^{1-T(\mu)} - 1}{\alpha-1} -2\mu+2\sum_{i=0}^\infty \sum_{j=0}^i \sum_{k=0}^\infty \omega_{i,j,k} \frac{\eta(j+\upsilon k+1)}{\phi} \\
&\quad\times\left\{B\left(\frac{1}{\phi}+1,\eta(j+\upsilon k+1)-\frac{1}{\phi}\right)-B_{Z_\mu} \left(\frac{1}{\phi}+1,\eta(j+\upsilon k+1)-\frac{1}{\phi}\right)\right\}
\end{aligned}
\label{Eq:81}
\end{equation}

Subsequently, the mean deviation about the median is given as
\begin{equation}
W_2 (X)=\E|X-M|=\int_0^\infty |X-M|f(x;\aleph)dx,
\label{Eq:82}
\end{equation}
\begin{equation}
W_2 (X)=-\mu+2\int_M^\infty xf(x;\aleph)dx
\label{Eq:83}
\end{equation}
After the beta transform, the final expression of the mean deviation about the median of the \APHBXII\ model is derived to be 
\begin{equation}
\begin{aligned}
W_2 (X)&=-\mu+2\sum_{i=0}^\infty \sum_{j=0}^i \sum_{k=0}^\infty \omega_{i,j,k} \frac{\eta(j+\upsilon k+1)}{\phi} \\
&\quad\times\left\{B\left(\frac{1}{\phi}+1,\eta(j+\upsilon k+1)-\frac{1}{\phi}\right)-B_{Z_M} \left(\frac{1}{\phi}+1,\eta(j+\upsilon k+1)-\frac{1}{\phi}\right)\right\}
\end{aligned}
\label{Eq:84}
\end{equation}

\subsubsection{Mean residual life of the APHBXII model}\label{ssubsec:3.5.5}

For a nonnegative random variate $X$ with $f(x)$, $F(x)$, and $S(x)=1-F(x)$, the mean residual life (MRL) at a given time $t$ is defined as $m(t)=\E[X-t | X>t]$, which can be written as
\begin{equation}
m(t)=\frac{1}{S(t)} \left[\int_t^\infty (x-t) f(x)dx\right]=\frac{1}{S(t)} \left[\int_t^\infty x f(x)dx-t\int_t^\infty f(x)dx\right]
\label{Eq:85}
\end{equation}
\begin{equation}
m(t)=\frac{1}{S(t)} \left[\int_t^\infty x f(x)dx-tS(t)\right]=\frac{1}{S(t)} \int_t^\infty x f(x)dx-t
\label{Eq:86}
\end{equation}
After the beta transform, the final expression of the mean residual life of the \APHBXII\ model is derived to be 
\begin{equation}
\begin{aligned}
m(t)&=\frac{1}{S(t)} \left(\sum_{i=0}^\infty \sum_{j=0}^i \sum_{k=0}^\infty \omega_{i,j,k} \frac{\eta(j+\upsilon k+1)}{\phi}\right. \\
&\quad\times \left.\left\{B\left(\frac{1}{\phi}+1,\eta(j+\upsilon k+1)-\frac{1}{\phi}\right)-B_{Z_t} \left(\frac{1}{\phi}+1,\eta(j+\upsilon k+1)-\frac{1}{\phi}\right)\right\}\right)-t
\end{aligned}
\label{Eq:87}
\end{equation}
where,
\begin{equation*}
\frac{1}{S(t)}=\frac{\alpha-1}{\alpha-\alpha^{1-T(x)}}
\end{equation*}

\subsubsection{The Average waiting time of APHBXII model}\label{subsubsec:3.5.6}

The average waiting time (AWT) of an item failed within an interval $[0,t]$ that follows APHBXII model is defined by
\begin{equation}
\bar{\mu} (t)=\E[t-X | X\leq t]=\frac{1}{F(t)} \int_0^t (t-x) f(x,\aleph)dx=t- \frac{1}{F(t)}\int_0^t x f(x,\aleph)dx
\label{Eq:88}
\end{equation}
The average waiting time of \APHBXII\ is given by
\begin{equation}
\label{Eq:89}
\begin{aligned}
\bar{\mu}(t)&=t-\frac{1}{F(t)} \left[\sum_{i=0}^\infty \sum_{j=0}^i \sum_{k=0}^\infty w_{i,j,k} \frac{\eta(j+\upsilon k+1)}{\phi} \right. \\
&\quad\times \left. B_{Z_t} \left(\frac{1}{\phi}+1,\eta(j+\upsilon k+1)-\frac{1}{\phi}\right)\right]
\end{aligned}
\end{equation}
where,
\begin{equation*}
\frac{1}{F(t)} =\frac{\alpha-1}{\alpha^{1-T(t)} -1}
\end{equation*}

\section{Maximum likelihood estimation of APHBXII distribution}\label{Sec:4} 

Suppose $X_1,X_2,...,X_n$ represent a collection of i.i.d. observations generated from the \APHBXII\ model ($\aleph$), then its likelihood function is given as 
\begin{equation}
L(x;;\aleph)=\prod_{i=1}^n \left(\frac{\log\alpha}{\alpha-1}\right) \frac{c^{1/\upsilon} \phi\eta x^{\phi-1} (1+x^\phi )^{-\eta}}{\{1-\bar{c}(1+x^\phi )^{-\eta\upsilon} \}^{1+1/\upsilon}} \alpha^{\left\{1-\frac{c^{1/\upsilon} (1+x^\phi )^{-\eta}}{(1-\bar{c}(1+x^\phi )^{-\eta\upsilon} )^{1/\upsilon} }\right\}}
\label{Eq:90}
\end{equation}
and, the log-likelihood function [$l(x;\aleph)=\log L(x;\aleph)=\sum_{i=1}^n \log f(x_i;\aleph)$] is given by
\begin{equation}
\begin{aligned}
l&=n\log\left(\frac{\log\alpha}{\alpha-1}\right)+\frac{n}{\upsilon}\log c+n\log\phi+n\log\eta+(\phi-1) \sum_{i=1}^n\log x_i \\
&\quad-\eta\sum_{i=1}^n\log(1+x_i^{\phi}) -\left(1+\frac{1}{\upsilon}\right) \sum_{i=1}^n\log\left\{1-\cbar(1+x_i^{\phi})^{-\eta\upsilon}\right\} \\
&\quad+\log\alpha\sum_{i=1}^n\left\{1-\frac{c^{1/\upsilon} (1+x_i^{\phi})^{-\eta}}{\left(1-\cbar(1+x_i^{\phi})^{-\eta\upsilon}\right)^{1/\upsilon}}\right\}
\end{aligned}
\label{Eq:91}
\end{equation}
Therefore, the partially differentiated log-likelihood functions with respect to $\upsilon$, $c$, $\alpha$, $\phi$, and $\eta$ are given below as
\begin{equation}
\partialderiv{l}{\alpha}=n\left(\frac{1}{\alpha\log\alpha}-\frac{1}{\alpha-1}\right)+\frac{1}{\alpha} \sum_{i=1}^n(1-T_i)
\label{Eq:92}
\end{equation}
\begin{equation}
\partialderiv{l}{\eta}=\frac{n}{\eta}-\sum_{i=1}^n\log A_i -\left(1+\frac{1}{\upsilon}\right)\sum_{i=1}^n\frac{\cbar\log A_i B_i}{D_i} +\log\alpha \sum_{i=1}^n T_i \log A_i \left(1+\frac{\cbar B_i}{D_i}\right)
\label{Eq:93}
\end{equation}
\begin{equation}
\partialderiv{l}{\phi}=\frac{n}{\phi}+\sum_{i=1}^n\log x_i -\eta\sum_{i=1}^n\partialderiv{}{\phi} (\log A_i)-\left(1+\frac{1}{\upsilon}\right) \sum_{i=1}^n\frac{1}{D_i} \partialderiv{D_i}{\phi}-\log\alpha \sum_{i=1}^n\partialderiv{T_i}{\phi}
\label{Eq:94}
\end{equation}
\begin{equation}
\partialderiv{l}{\upsilon}=-\frac{n}{\upsilon^2}\log c - \sum_{i=1}^n\log c - \sum_{i=1}^n\left[\partialderiv{}{\upsilon}\left(1+\frac{1}{\upsilon}\right)\log D_i\right] +\log\alpha \sum_{i=1}^n\left(-\partialderiv{T_i}{\upsilon}\right)
\label{Eq:95}
\end{equation}
\begin{equation}
\partialderiv{l}{c}=\frac{n}{\upsilon c}-\left(1+\frac{1}{\upsilon}\right) \sum_{i=1}^n\frac{B_i}{D_i} -\log\alpha \sum_{i=1}^n T_i \left(\frac{1}{\upsilon c}+\frac{B_i}{\upsilon D_i}\right)
\label{Eq:96}
\end{equation}
where, $A_i=1+x_i^{\phi}$, $B_i=A_i^{-\eta\upsilon}$, $D_i=1-\cbar B_i$, $T_i=c^{1/\upsilon} A_i^{-\eta} D_i^{-1/\upsilon}$.

However, due to the presence of nested logarithmic terms, powers, and cross-product interactions in the score equations, the resulting system of nonlinear likelihood equations cannot be solved analytically when set to zero. This implies that the \APHBXII\ distribution does not have a closed-form expression for the maximum likelihood estimates of the parameters. Therefore, we obtained an approximate value of the parameter estimates numerically by maximizing the log-likelihood function, which we compute using the R statistical software.

\section{Simulation study}\label{Sec:5} 

In this section, we conducted a Monte Carlo simulation study to investigate the finite-sample behavior of the maximum likelihood estimators (MLEs) of the \APHBXII\ distribution. Artificial observations were generated using the inverse-transform sampling method. Based on the quantile representation derived in Eq. (\ref{Eq:18}), a large pseudo-population of size $N=20,000$ was first generated for each parameter configuration. Monte Carlo experiments were then performed by drawing simple random samples of sizes $n \in \{20, 50, 100, 150, 200, 250, \\ 350\}$ from each population. For every sample size and parameter setting, $R=1000$ replications were executed. The parameters were estimated using the MLE technique. The log-likelihood maximization was carried out using the Nelder–Mead simplex method algorithm implemented in the \texttt{AdequacyModel} package in R, where the PDF and CDF of the \APHBXII\ model were supplied.

Four sets of true parameter values were considered for the data-generating process: 
\begin{enumerate}
    \item $\upsilon=2.5, c=0.8, \alpha=1.8, \phi=1.0, \eta=2.5$;
    \item $\upsilon=1.5, c=0.6, \alpha=2.2, \phi=1.2, \eta=1.8$;
    \item $\upsilon=3.0, c=0.4, \alpha=1.5, \phi=0.8, \eta=3.0$; and
    \item $\upsilon=2.0, c=0.7, \alpha=2.5, \phi=1.5, \eta=2.0$.
\end{enumerate}
For each combination of parameter set and sample size, the Monte Carlo estimates of the absolute bias (AB), standard error (SE), and mean squared error (MSE) of the MLEs were computed. \hyperref[Table:3]{Table 3} report the results for the parameters $\upsilon$, $c$, $\alpha$, $\phi$, and $\eta$. Our findings revealed that as the sample sizes $n$ increases, the standard error and the mean square error of the estimates decayed towards zero. This indicates that larger sample sizes produce more reliable estimates, which in turn reduces bias and stochastic errors. This implies that the maximum likelihood estimator is consistent for the parameters of the \APHBXII\ distribution.

\begin{sidewaystable*}[p!]
    \centering
     \footnotesize
    \caption{Absolute bias (AB), standard error (SE) and mean squared error (MSE) for MLEs of $\upsilon$, c, $\alpha$, $\phi$, and $\eta$ with varying values.}
    \label{Table:3}
    \begin{tabular}{@{}l *{5}{ccc}@{}}
    \toprule
    
    \multirow{2}{*}{n} & \multicolumn{3}{c}{$\upsilon=2.5$} & \multicolumn{3}{c}{c=0.8} & \multicolumn{3}{c}{$\alpha=1.8$} & \multicolumn{3}{c}{$\phi=1.0$} & \multicolumn{3}{c}{$\eta=2.5$} \\
    \cmidrule(lr){2-4} \cmidrule(lr){5-7} \cmidrule(lr){8-10} \cmidrule(lr){11-13} \cmidrule(lr){14-16}
     & AB & SE & MSE & AB & SE & MSE & AB & SE & MSE & AB & SE & MSE & AB & SE & MSE \\
    \midrule
    20 & 2.8662 & 5.8269 & 42.1276 & 2.3762 & 9.6111 & 97.9100 & 1.4671 & 6.6197 & 45.9206 & 0.4795 & 0.6712 & 0.6798 & 1.2012 & 7.1557 & 52.5863 \\
    50 & 2.3159 & 5.4448 & 34.9787 & 1.6186 & 5.3843 & 31.5805 & 1.7123 & 6.4811 & 44.8925 & 0.2677 & 0.4992 & 0.3206 & 0.7430 & 5.3249 & 28.8774 \\
    100 & 2.1569 & 5.5563 & 35.4935 & 0.8687 & 3.4083 & 12.3592 & 1.5138 & 6.3460 & 42.5234 & 0.1713 & 0.3596 & 0.1585 & 0.1300 & 2.2626 & 5.1310 \\
    150 & 1.6387 & 4.3356 & 21.4634 & 0.7028 & 2.6345 & 7.4273 & 1.0973 & 5.6000 & 32.5324 & 0.1214 & 0.2964 & 0.1025 & 0.0895 & 1.7864 & 3.1959 \\
    200 & 1.5288 & 4.0656 & 18.8496 & 0.4867 & 2.1829 & 4.9972 & 0.9560 & 3.8594 & 15.7939 & 0.1057 & 0.2649 & 0.0813 & 0.0287 & 2.0299 & 4.1170 \\
    250 & 1.5947 & 4.3376 & 21.3385 & 0.4225 & 2.2534 & 5.2514 & 0.8703 & 3.3970 & 12.2856 & 0.0821 & 0.2352 & 0.0620 & 0.0869 & 1.6840 & 2.8404 \\
    350 & 1.3414 & 3.6716 & 15.2666 & 0.1935 & 1.1537 & 1.3671 & 0.9094 & 3.7755 & 15.0672 & 0.0817 & 0.2015 & 0.0472 & 0.0328 & 0.9459 & 0.8949 \\
    \midrule
    
    \multirow{2}{*}{n} & \multicolumn{3}{c}{$\upsilon=1.5$} & \multicolumn{3}{c}{c=0.6} & \multicolumn{3}{c}{$\alpha=2.2$} & \multicolumn{3}{c}{$\phi=1.2$} & \multicolumn{3}{c}{$\eta=1.8$} \\
    \cmidrule(lr){2-4} \cmidrule(lr){5-7} \cmidrule(lr){8-10} \cmidrule(lr){11-13} \cmidrule(lr){14-16}
     & AB & SE & MSE & AB & SE & MSE & AB & SE & MSE & AB & SE & MSE & AB & SE & MSE \\
    \midrule
    20 & 1.6365 & 4.2048 & 20.3362 & 3.0684 & 10.0803 & 110.9001 & 2.0928 & 7.6699 & 63.1328 & 0.5142 & 0.7814 & 0.8743 & 2.5638 & 8.8311 & 84.4638 \\
    50 & 1.3294 & 3.6816 & 15.3064 & 1.0340 & 4.4112 & 20.5072 & 2.0299 & 6.2353 & 42.9579 & 0.3065 & 0.5575 & 0.4044 & 0.7352 & 4.1950 & 18.1196 \\
    100 & 1.4390 & 3.4992 & 14.3023 & 0.2248 & 2.1222 & 4.5494 & 1.4079 & 4.2714 & 20.2078 & 0.2428 & 0.4335 & 0.2467 & 0.0480 & 1.8883 & 3.5644 \\
    150 & 1.2923 & 3.3693 & 13.0107 & 0.1586 & 1.8073 & 3.2883 & 1.1146 & 3.5184 & 13.6086 & 0.1731 & 0.3330 & 0.1407 & 0.0007 & 1.5928 & 2.5345 \\
    200 & 1.2095 & 3.1631 & 11.4580 & 0.0299 & 1.0401 & 1.0817 & 1.2201 & 3.1452 & 11.3709 & 0.1416 & 0.2930 & 0.1058 & 0.1214 & 1.5079 & 2.2863 \\
    250 & 1.0702 & 2.9869 & 10.0577 & 0.0697 & 0.6961 & 0.4889 & 1.0591 & 3.1286 & 10.8995 & 0.1165 & 0.2564 & 0.0792 & 0.1509 & 0.8065 & 0.6726 \\
    350 & 1.0211 & 2.8936 & 9.4070 & 0.1075 & 0.5944 & 0.3645 & 0.8953 & 2.5272 & 7.1819 & 0.0982 & 0.2201 & 0.0581 & 0.1911 & 0.6453 & 0.4526 \\
    \midrule
    
    \multirow{2}{*}{n} & \multicolumn{3}{c}{$\upsilon=3.0$} & \multicolumn{3}{c}{c=0.4} & \multicolumn{3}{c}{$\alpha=1.5$} & \multicolumn{3}{c}{$\phi=0.8$} & \multicolumn{3}{c}{$\eta=3.0$} \\
    \cmidrule(lr){2-4} \cmidrule(lr){5-7} \cmidrule(lr){8-10} \cmidrule(lr){11-13} \cmidrule(lr){14-16}
     & AB & SE & MSE & AB & SE & MSE & AB & SE & MSE & AB & SE & MSE & AB & SE & MSE \\
    \midrule
    20 & 3.2183 & 7.2948 & 63.5068 & 1.5634 & 7.6124 & 60.3224 & 1.3628 & 5.1145 & 27.9839 & 0.3180 & 0.4403 & 0.2948 & 1.1871 & 8.6261 & 75.7288 \\
    50 & 2.4699 & 5.5839 & 37.2479 & 0.5120 & 2.4823 & 6.4176 & 1.4731 & 5.0955 & 28.1064 & 0.1684 & 0.2996 & 0.1180 & 0.0728 & 3.5903 & 12.8820 \\
    100 & 2.0321 & 5.1726 & 30.8585 & 0.2367 & 1.5500 & 2.4561 & 1.1754 & 4.4066 & 20.7807 & 0.1064 & 0.2240 & 0.0614 & 0.1736 & 1.7292 & 3.0174 \\
    150 & 1.9723 & 5.1240 & 30.1187 & 0.1717 & 0.9445 & 0.9206 & 0.8723 & 4.1364 & 17.8530 & 0.0754 & 0.1817 & 0.0387 & 0.2424 & 0.9966 & 1.0510 \\
    200 & 1.6721 & 4.9058 & 26.8383 & 0.1611 & 0.8531 & 0.7529 & 0.8190 & 2.6929 & 7.9153 & 0.0509 & 0.1634 & 0.0293 & 0.1442 & 1.1026 & 1.2352 \\
    250 & 1.8026 & 4.9517 & 27.7436 & 0.1003 & 0.6607 & 0.4461 & 0.7092 & 2.3805 & 6.1642 & 0.0493 & 0.1483 & 0.0244 & 0.1813 & 0.8741 & 0.7961 \\
    350 & 1.5630 & 4.3038 & 20.9471 & 0.0924 & 0.5860 & 0.3516 & 0.4565 & 1.8826 & 3.7492 & 0.0339 & 0.1246 & 0.0167 & 0.2172 & 0.5760 & 0.3786 \\
    \midrule
    
    \multirow{2}{*}{n} & \multicolumn{3}{c}{$\upsilon=2.0$} & \multicolumn{3}{c}{c=0.7} & \multicolumn{3}{c}{$\alpha=2.5$} & \multicolumn{3}{c}{$\phi=1.5$} & \multicolumn{3}{c}{$\eta=2.0$} \\
    \cmidrule(lr){2-4} \cmidrule(lr){5-7} \cmidrule(lr){8-10} \cmidrule(lr){11-13} \cmidrule(lr){14-16}
     & AB & SE & MSE & AB & SE & MSE & AB & SE & MSE & AB & SE & MSE & AB & SE & MSE \\
    \midrule
    20 & 1.6919 & 4.4007 & 22.2048 & 2.5931 & 8.3842 & 76.9317 & 1.7054 & 7.7432 & 62.7910 & 0.7024 & 1.0325 & 1.5582 & 1.4626 & 6.0755 & 39.0050 \\
    50 & 1.4268 & 4.0228 & 18.2021 & 1.0708 & 5.1836 & 27.9881 & 2.2129 & 7.4366 & 60.1418 & 0.4357 & 0.7504 & 0.7524 & 0.5226 & 4.1501 & 17.4783 \\
    100 & 1.2273 & 3.7774 & 15.7611 & 0.5547 & 2.7480 & 7.8518 & 1.6391 & 4.9169 & 26.8387 & 0.2777 & 0.5454 & 0.3742 & 0.0583 & 1.8004 & 3.2414 \\
    150 & 1.1747 & 4.3615 & 20.3835 & 0.3449 & 2.3467 & 5.6204 & 1.7639 & 7.1849 & 54.6816 & 0.2046 & 0.4419 & 0.2370 & 0.0593 & 2.0927 & 4.3787 \\
    200 & 0.9296 & 3.2429 & 11.3700 & 0.1538 & 1.3556 & 1.8593 & 1.6566 & 4.5811 & 23.7100 & 0.1571 & 0.3870 & 0.1743 & 0.0508 & 1.3288 & 1.7665 \\
    250 & 0.9049 & 3.2266 & 11.2197 & 0.0628 & 1.6301 & 2.6586 & 1.4000 & 4.4085 & 21.3754 & 0.1475 & 0.3481 & 0.1428 & 0.1116 & 0.9468 & 0.9081 \\
    350 & 0.8958 & 3.2741 & 11.5114 & 0.0196 & 0.9085 & 0.8250 & 0.9208 & 2.6507 & 7.8669 & 0.1128 & 0.2830 & 0.0927 & 0.1326 & 0.7495 & 0.5788 \\
    \bottomrule
    \end{tabular}
\end{sidewaystable*}

\section{Application of APHBXII distribution to lifetime data}\label{Sec:6}

In this section, we examine the performance, flexibility and practical value of the \APHBXII\ model using three lifetime datasets.

The first dataset is made up of 101 observations obtained from the failure times of stress-rupture life (in hours) of Kevlar 49/epoxy strands, subjected to a constant sustained pressure at the 90\% stress level, as reported by \cite{Andrews2012data}. The data are as follows: 
0.02, 0.09, 0.18, 0.19, 0.20, 0.23, 0.24, 0.24, 0.29, 0.34, 0.35, 0.36, 0.38, 0.40, 0.42, 0.43, 0.52, 0.54, 0.56, 0.60, 0.60, 0.63, 0.65, 0.10, 0.12, 0.67, 0.68, 0.72, 0.72, 0.72, 0.73, 0.03, 0.04, 0.79, 0.79, 0.80, 0.80, 0.83, 0.85, 0.90, 0.92, 0.10, 0.95, 0.99, 1.00, 1.01, 1.02, 1.03, 1.05, 1.10, 1.10, 0.07, 1.11, 1.15, 1.18, 1.20, 1.29, 1.31, 1.33, 1.34, 1.40, 0.01, 1.43, 0.11, 0.11, 1.45, 1.50, 1.51, 1.52, 0.06, 0.08, 1.53, 1.54, 1.54, 0.03, 1.55, 1.58, 1.60, 1.63, 1.64, 1.80, 0.01, 1.80, 1.81, 2.02, 2.05, 0.09, 0.13, 2.14, 2.17, 2.33, 0.02, 0.05, 3.03, 3.03, 0.07, 3.34, 4.20, 4.69, 0.02, 7.89.

The second set of data, as reported by \cite{Lee2003statistical}, represents the remission times (in months) of a random sample of 128 bladder cancer patients. The data are as follows: 
23.63, 0.40, 2.23, 1.40, 3.02, 4.34, 5.71, 7.93, 11.79, 18.10, 1.46, 4.40, 5.85, 8.26, 11.98, 19.13, 1.76, 3.25, 4.50, 6.25, 8.37, 12.02, 2.02, 3.31, 4.51, 6.54, 8.53, 12.03, 20.28, 2.02, 3.36, 6.76, 12.07, 21.73, 2.07, 3.36, 6.93, 8.65, 12.63, 22.69, 0.08, 2.09, 3.48, 4.87, 6.94, 8.66, 13.11, 0.20, 3.52, 4.98, 6.97, 9.02, 13.29, 2.26, 3.57, 5.06, 7.09, 9.22, 13.80, 25.74, 0.50, 2.46, 3.64, 5.09, 7.26, 9.47, 14.24, 25.82, 0.51, 2.54, 3.70, 5.17, 7.28, 10.06, 14.77, 26.31, 0.81, 2.62, 3.82, 5.32, 7.32, 10.06, 14.77, 32.15, 2.64, 3.88, 5.32, 7.39, 10.34, 14.83, 34.26, 0.90, 2.69, 4.18, 5.34, 7.59, 10.66, 15.96, 36.66, 1.05, 2.69, 4.23, 5.41, 7.62, 10.75, 16.62, 43.01, 1.19, 2.75, 4.26, 5.41, 7.63, 17.12, 46.12, 1.26, 2.83, 4.33, 5.49, 7.66, 11.25, 17.14, 79.05, 1.35, 2.87, 5.62, 7.87, 11.64, 17.36.

The third dataset used by \cite{Aarset1987identify}, consists of 50 data points that represent the observed time until failure for a device placed on a life test at time zero. The data are as follows: 
0.1, 0.2, 1, 1, 1, 1, 1, 2, 3, 6, 7, 11, 12, 18, 18, 18, 18, 18, 21, 32, 36, 40, 45, 46, 47, 50, 55, 60, 63, 63, 67, 67, 67, 67, 72, 75, 79, 82, 82, 83, 84, 84, 84, 85, 85, 85, 85, 85, 86, 86.

The exploratory data analysis (EDA) results of the three sets of data are recorded and shown in their descriptive analysis, as reported in \hyperref[Table:4]{Table 4}, boxplots, kernel density plots, violin plots, and the TTT plots in \hyperref[Fig:3]{Figs. 3}, \ref{Fig:4}, \ref{Fig:5}, and \ref{Fig:6}, respectively. The Kevlar49/Epoxy data and the cancer data both exhibit the same general pattern, with most observations clustered at the lower end and a few unusually large values stretching the distribution to the right. This creates a strongly positively skewed and leptokurtic shape, which is also supported by their TTT plots that lie close to the diagonal, indicating an almost constant failure rate. In contrast, the lifetime failure data behave very differently. The values are much more widely spread, the distribution is nearly symmetric, and the density plot shows it is bimodal. Its TTT plot departs noticeably from the straight line and forms the characteristic pattern of a bathtub-shaped hazard, suggesting the presence of more than one underlying failure mechanism within the data. As seen in \hyperref[Fig:7]{Fig. 7}, each of the five models reflects the strong right-skew of the datasets, with APHBXII giving the best overall match to the empirical densities. Similarly, \hyperref[Fig:8]{Fig. 8} indicates substantial agreement in the cumulative plots, with APHBXII and MOBXII tracing the empirical CDFs and capturing the heavier concentration near the lower values.

\begin{table*}[ht]
    \centering
    \caption{Descriptive statistics for the data sets.}
    \label{Table:4}
    \begin{tabular}{@{}l c c c@{}}
    \toprule
    Statistic & Set I & Set II & Set III \\
    \midrule
    Minimum      & 0.010 & 0.080 & 0.100 \\
    1st Quartile & 0.240 & 3.348 & 13.500 \\
    Median       & 0.800 & 6.395 & 48.500 \\
    Mean         & 1.025 & 9.366 & 45.686 \\
    3rd Quartile & 1.450 & 11.838 & 81.250 \\
    Maximum      & 7.890 & 79.050 & 86.000 \\
    Variance     & 1.253 & 110.425 & 1078.153 \\
    Skewness     & 3.002 & 3.287 & -0.138 \\
    Kurtosis     & 16.709 & 18.483 & 1.414 \\
    \bottomrule
    \end{tabular}
\end{table*}

\begin{figure}[ht!]
    \centering
    \includegraphics[width=\textwidth]{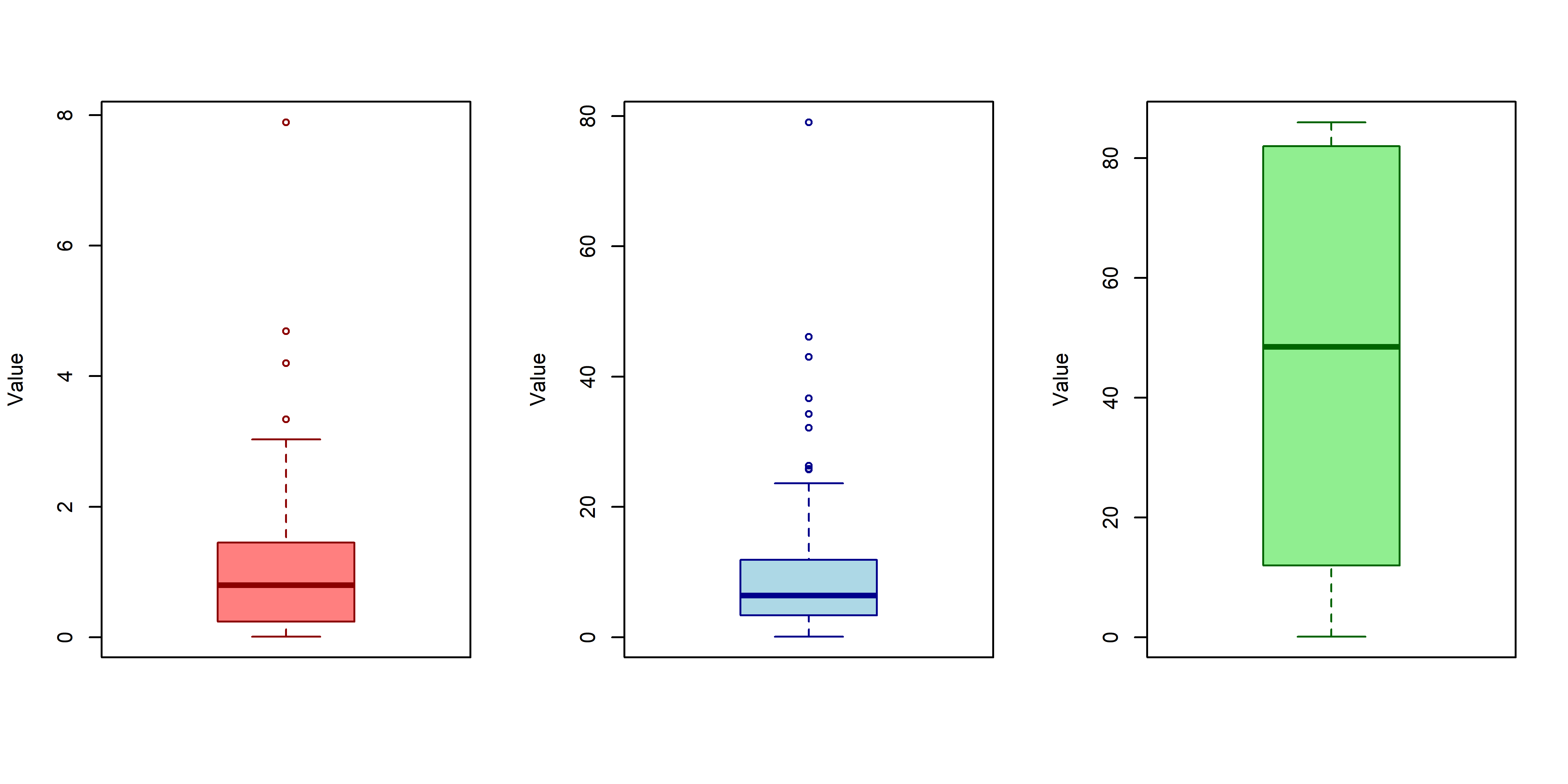}
    \caption{Boxplots for the datasets. The figure compares the distributional spread, skewness, and outliers present in the Kevlar/epoxy, cancer, and lifetime datasets.}
    \label{Fig:3}
\end{figure}

\begin{figure}[ht!]
    \centering
    \includegraphics[width=\textwidth]{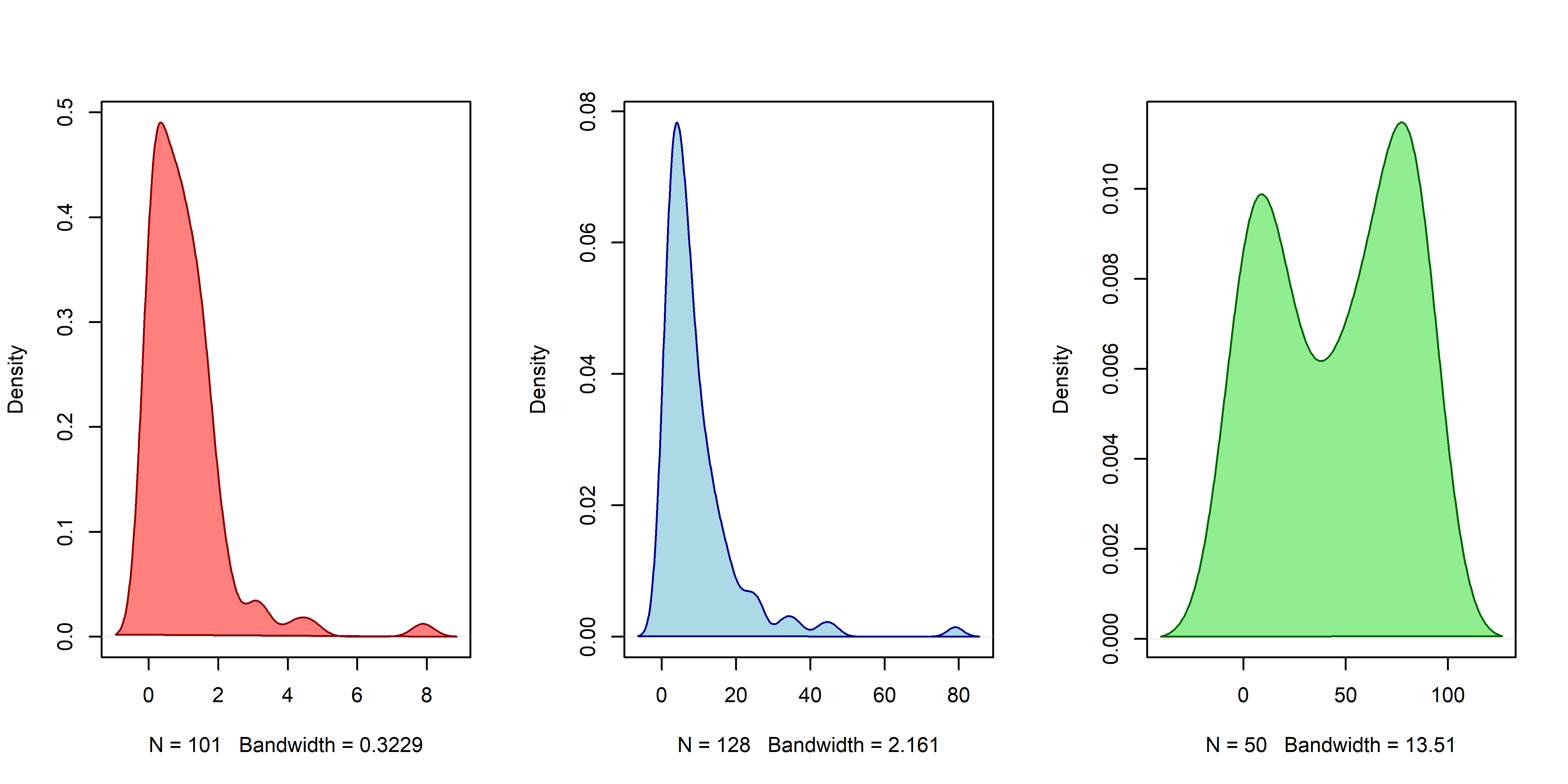}
    \caption{Kernel density plots for the datasets. Smoothed density curves are shown for the Kevlar/epoxy, cancer, and lifetime data, highlighting differences in their empirical distributions.}
    \label{Fig:4}
\end{figure}

\begin{figure}[ht!]
    \centering
    \includegraphics[width=\textwidth]{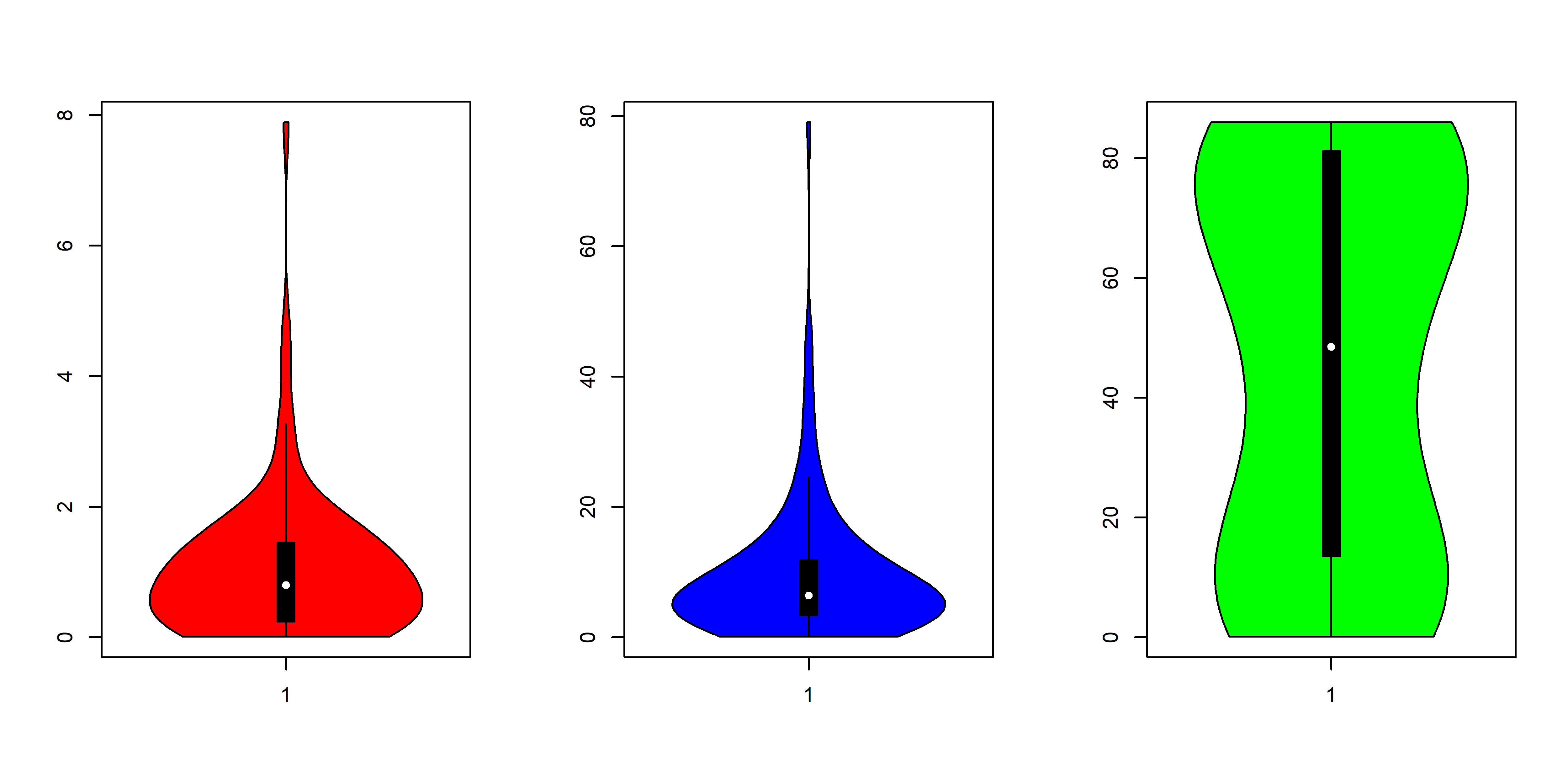} 
    \caption{Violin plots for the datasets. The violin diagrams represent the empirical distributions, combining boxplot summaries with kernel density shapes.}
    \label{Fig:5}
\end{figure}

\begin{figure}[ht!]
    \centering
    \includegraphics[width=\textwidth]{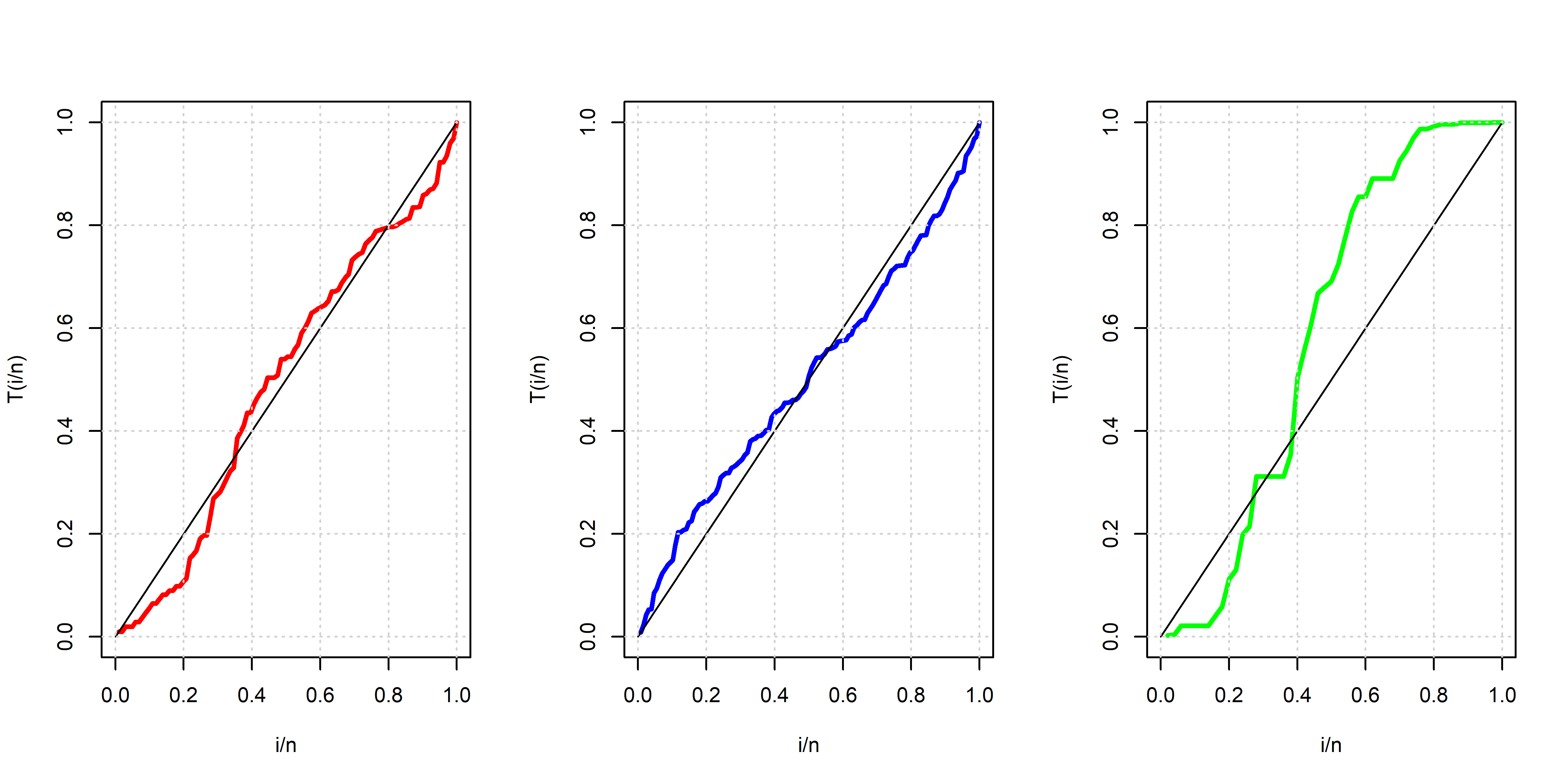}
    \caption{TTT (Total Time on Test) plots for the datasets. The TTT plots indicate the underlying ageing behavior (increasing, decreasing, or bathtub hazard tendencies) of each dataset.}
    \label{Fig:6}
\end{figure}

\begin{figure}[ht!]
    \centering
    \includegraphics[width=\textwidth]{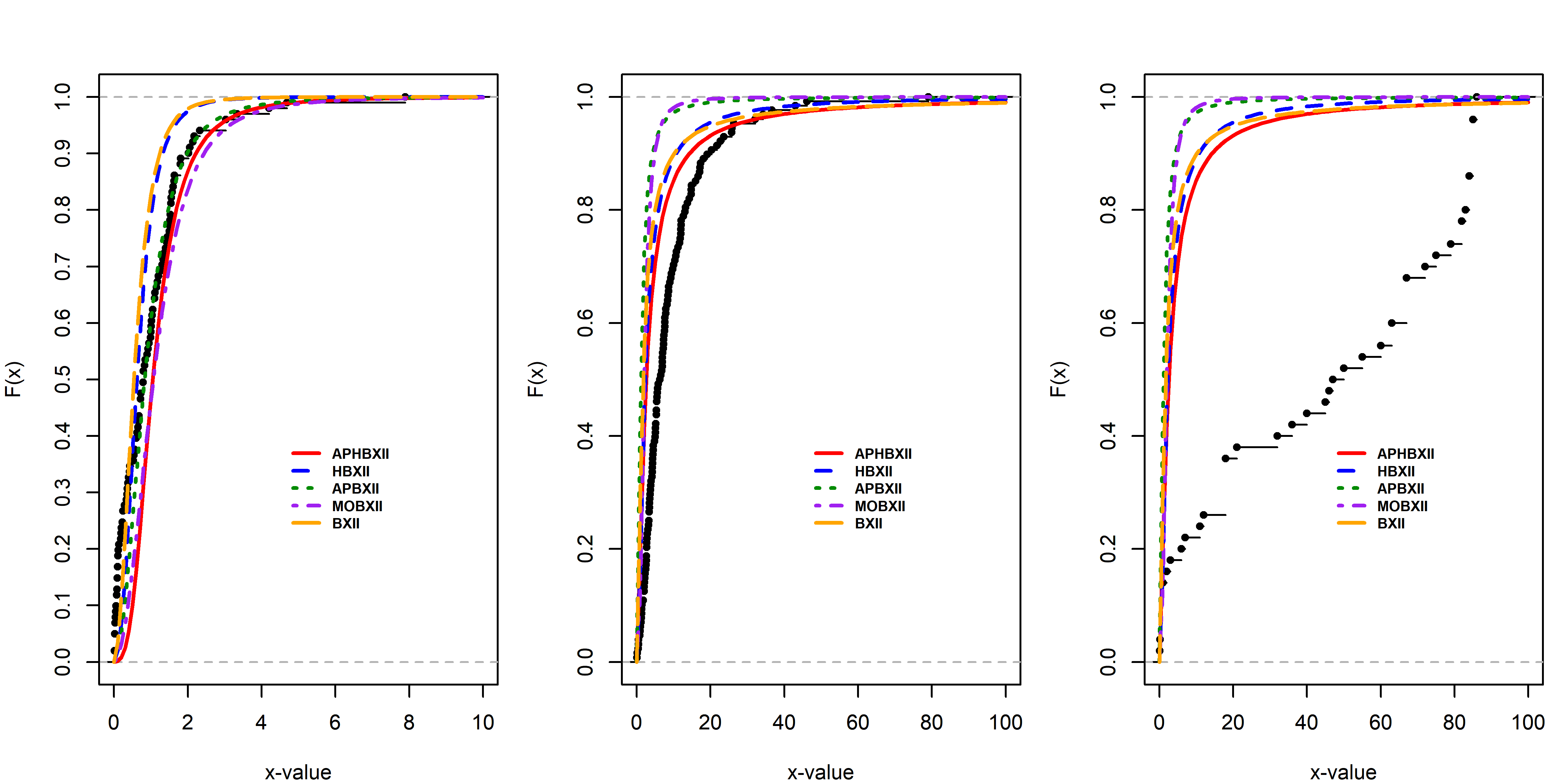}
    \caption{Estimated cumulative distribution functions for the datasets. The empirical CDFs are plotted alongside fitted model-based CDFs to assess goodness of fit.}
    \label{Fig:7}
\end{figure}

\begin{figure}[ht!]
    \centering
    \includegraphics[width=\textwidth]{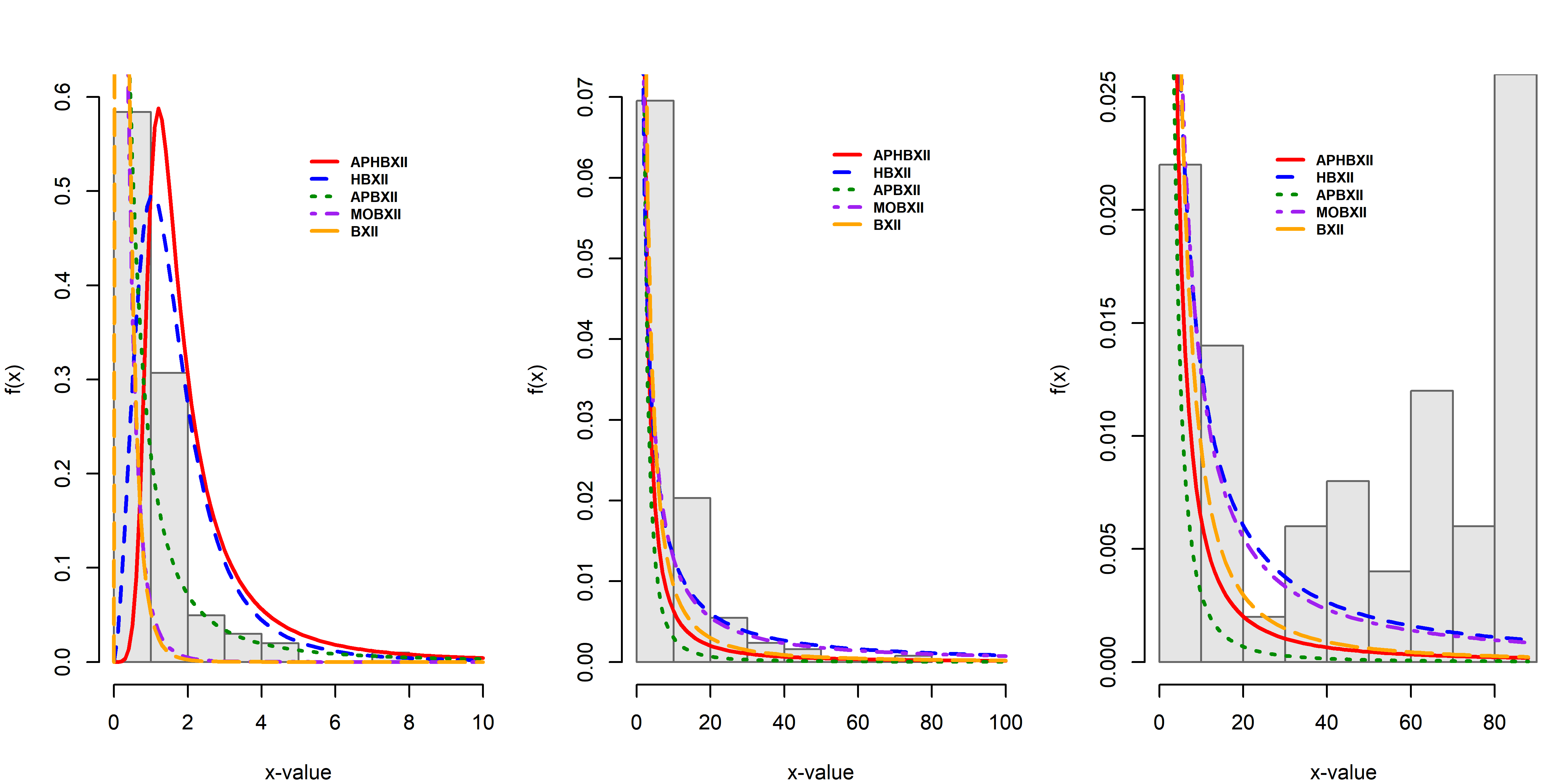}
    \caption{Estimated probability density functions for the datasets. The fitted PDFs are shown for each dataset, revealing how well the models capture the observed distributional shapes.}
    \label{Fig:8}
\end{figure}

Furthermore, we evaluate how well the \APHBXII\ distribution fits our set of data relative to four other benchmark models, which include Harris Burr XII (\HBXII) by \cite{Sivadas2022harris}, Alpha Power Burr XII (\APBXII) by \cite{Gbenga2022alpha}, Marshall-Olkin Burr XII (\MOBXII) by \cite{Al-Saiari2014marshall}, and the Burr XII (\BXII) by \cite{Burr1942cumulative}. The probability density functions of the competing models considered are, respectively, given by

\begin{equation}
f_{\HBXII} (x)=\frac{c^{1/\upsilon} \phi\eta x^{(\phi-1)} (1+x^{\phi})^{-\eta-1}}{\left(1-\cbar(1+x^{\phi})^{-\eta\upsilon}\right)^{1+1/\upsilon}}, \qquad x>0
\label{Eq:97} 
\end{equation}

\begin{equation}
f_{\APBXII} (x)=\left(\frac{\log\alpha}{\alpha-1}\right)\phi\eta x^{(\phi-1)} (1+x^{\phi})^{-\eta-1} \alpha^{1-(1+x^{\phi})^{-\eta}}, \qquad x>0
\label{Eq:98} 
\end{equation}

\begin{equation}
f_{\MOBXII} (x)=\frac{c\phi\eta x^{(\phi-1)} (1+x^{\phi})^{-\eta-1}}{\left(1-\cbar(1+x^{\phi})^{-\eta}\right)^2}, \qquad x>0
\label{Eq:99} 
\end{equation}

\begin{equation}
f_{\BXII} (x)=\phi\eta x^{(\phi-1)} (1+x^{\phi})^{-\eta-1}, \qquad x>0
\label{Eq:100} 
\end{equation}

For each model, we estimate the value of the parameters using the MLE technique and compute their corresponding log-likelihood values at the MLEs with the simulated annealing (SANN) optimization algorithm. \hyperref[Table:5]{Table 5}, \hyperref[Table:6]{6}, and \hyperref[Table:7]{7} show MLEs with SEs (in parentheses) for our sets of data. Furthermore, we compare the models using nine established goodness-of-fit measures, which include deviance ($\vartheta_1$ or $-2l$), Akaike information criterion ($\vartheta_2$), Bayesian information criterion ($\vartheta_3$), Hannan-Quinn information criterion ($\vartheta_4$), consistent Akaike information criterion ($\vartheta_5$), Cramer-von Mises ($\vartheta_6$) test, Anderson-Darling ($\vartheta_7$) test, Kolmogorov-Smirnov ($\vartheta_8$) test, and the K-S p-value ($\vartheta_9$) statistic as shown in \hyperref[Table:8]{Table 8}, \hyperref[Table:9]{9}, and \hyperref[Table:10]{10}, where the \APHBXII\ model with the lowest values of the first eight metrics and the highest K-S p-value statistic in most cases was considered the best fit. Likewise, we conducted the likelihood ratio tests, and their results, as shown in \hyperref[Table:11]{Table 11}, \hyperref[Table:12]{12}, and \hyperref[Table:13]{13}, revealed that the proposed \APHBXII\ model provides a significantly better fit than its nested competitors.

\begin{table*}[ht!]
    \centering
     \footnotesize
    \caption{MLEs and standard error (in parenthesis) for Kevlar/Epoxy data}
    \label{Table:5}
    \begin{tabular}{@{}l c c c c c@{}}
    \toprule
    Model & $\upsilon$ & c & $\alpha$ & $\phi$ & $\eta$ \\
    \midrule
    APHBXII & 0.003(0.004) & 8.394(3.300) & 10.552(13.446) & 1.926(0.289) & 1.419(0.249) \\
    HBXII   & 0.395(0.200) & 10.314(10.847) & -- & 0.799(0.227) & 6.447(3.039) \\
    APBXII  & -- & -- & 5.302(4.157) & 0.991(0.130) & 2.361(0.384) \\
    MOBXII  & -- & 7.668(6.425) & -- & 0.786(0.171) & 3.837(1.094) \\
    BXII    & -- & -- & -- & 1.174(0.098) & 1.632(0.164) \\
    \bottomrule
    \end{tabular}
\end{table*}

\begin{table*}[ht!]
    \centering
     \footnotesize
    \caption{MLEs and standard error (in parenthesis) for cancer data}
    \label{Table:6}
    \begin{tabular}{@{}l c c c c c@{}}
    \toprule
    Model & $\upsilon$ & c & $\alpha$ & $\phi$ & $\eta$ \\
    \midrule
    APHBXII & 0.524(0.308) & 14.692(9.487) & 3.957(4.064) & 1.293(0.355) & 1.945(0.700) \\
    HBXII   & 0.644(0.362) & 23.021(7.835) & -- & 1.438(0.343) & 1.561(0.654) \\
    APBXII  & -- & -- & 18.177(6.082) & 1.989(0.387) & 0.488(0.104) \\
    MOBXII  & -- & 19.865(6.411) & -- & 1.413(0.246) & 1.177(0.250) \\
    BXII    & -- & -- & -- & 2.336(0.349) & 0.232(0.039) \\
    \bottomrule
    \end{tabular}
\end{table*}

\begin{table*}[ht!]
    \centering
     \footnotesize
    \caption{MLEs and standard error (in parenthesis) for lifetime data}
    \label{Table:7}
    \begin{tabular}{@{}l c c c c c@{}}
    \toprule
    Model & $\upsilon$ & c & $\alpha$ & $\phi$ & $\eta$ \\
    \midrule
    APHBXII & 0.276(0.102) & 17.794(8.363) & 8.478(5.884) & 0.782(0.191) & 2.952(0.967) \\
    HBXII & 0.514(0.162) & 18.175(7.412) & -- & 0.986(0.217) & 1.376(0.384) \\
    APBXII & -- & -- & 21.609(12.182) & 0.972(0.287) & 0.565(0.186) \\
    MOBXII & -- & 16.147(6.731) & -- & 0.822(0.174) & 1.070(0.267) \\
    BXII & -- & -- & -- & 1.261(0.322) & 0.245(0.070) \\
    \bottomrule
    \end{tabular}
\end{table*}

\begin{table*}[ht!]
    \centering
     \footnotesize
    \caption{Goodness-of-fit metrics for the Kevlar 49/Epoxy data}
    \label{Table:8}
    \begin{tabular}{@{}l ccccc ccccc@{}}
    \toprule
    Model & $\vartheta_1$ & $\vartheta_2$ & $\vartheta_3$ & $\vartheta_4$ & $\vartheta_5$ & $\vartheta_6$ & $\vartheta_7$ & $\vartheta_8$ & $\vartheta_9$ \\
    \midrule
    APHBXII & 198.712 & 208.712 & 221.788 & 214.006 & 209.344 & 0.145 & 0.806 & 0.105 & 0.211 \\
    HBXII   & 203.976 & 211.976 & 222.436 & 216.211 & 212.393 & 0.145 & 0.849 & 0.082 & 0.513 \\
    APBXII  & 212.415 & 218.415 & 226.260 & 221.591 & 218.662 & 0.342 & 1.867 & 0.105 & 0.219 \\
    MOBXII  & 207.519 & 213.519 & 221.364 & 216.695 & 213.766 & 0.226 & 1.269 & 0.091 & 0.370 \\
    BXII    & 217.096 & 221.096 & 226.326 & 223.213 & 221.218 & 0.440 & 2.386 & 0.136 & 0.047 \\
    \bottomrule
    \end{tabular}
\end{table*}

\begin{table*}[ht!]
    \centering
     \footnotesize
    \caption{Goodness-of-fit metrics for the cancer data}
    \label{Table:9}
    \begin{tabular}{@{}l ccccc ccccc@{}}
    \toprule
    Model & $\vartheta_1$ & $\vartheta_2$ & $\vartheta_3$ & $\vartheta_4$ & $\vartheta_5$ & $\vartheta_6$ & $\vartheta_7$ & $\vartheta_8$ & $\vartheta_9$ \\
    \midrule
    APHBXII & 819.195 & 829.195 & 843.456 & 834.989 & 829.687 & 0.014 & 0.096 & 0.037 & 0.995 \\
    HBXII   & 821.266 & 829.266 & 840.674 & 833.901 & 829.591 & 0.025 & 0.185 & 0.055 & 0.831 \\
    APBXII  & 855.073 & 861.073 & 869.629 & 864.549 & 861.267 & 0.315 & 2.014 & 0.147 & 0.008 \\
    MOBXII  & 822.580 & 828.580 & 837.136 & 832.056 & 828.773 & 0.038 & 0.275 & 0.050 & 0.906 \\
    BXII    & 907.034 & 911.034 & 916.738 & 913.351 & 911.130 & 0.748 & 4.547 & 0.251 & 1.93E-07 \\
    \bottomrule
    \end{tabular}
\end{table*}

\begin{table*}[ht!]
    \centering
     \footnotesize
    \caption{Goodness-of-fit metrics for the device lifetime data}
    \label{Table:10}
    \begin{tabular}{@{}l ccccc ccccc@{}}
    \toprule
    Model & $\vartheta_1$ & $\vartheta_2$ & $\vartheta_3$ & $\vartheta_4$ & $\vartheta_5$ & $\vartheta_6$ & $\vartheta_7$ & $\vartheta_8$ & $\vartheta_9$ \\
    \midrule
    APHBXII & 491.749 & 501.749 & 511.309 & 505.389 & 503.112 & 0.589 & 3.501 & 0.226 & 0.012 \\
    HBXII   & 505.343 & 513.343 & 520.991 & 516.256 & 514.232 & 0.732 & 4.194 & 0.240 & 0.006 \\
    APBXII  & 524.138 & 530.138 & 535.874 & 532.323 & 530.660 & 0.920 & 5.094 & 0.270 & 0.001 \\
    MOBXII  & 505.677 & 511.677 & 517.413 & 513.861 & 512.198 & 0.732 & 4.215 & 0.229 & 0.011 \\
    BXII    & 544.729 & 548.729 & 552.553 & 550.185 & 548.984 & 1.093 & 5.850 & 0.333 & 2.97E-05 \\
    \bottomrule
    \end{tabular}
\end{table*}

\begin{table*}[ht!]
    \centering
     \footnotesize
    \caption{Likelihood ratio statistics for the Kevlar 49/Epoxy data}
    \label{Table:11}
    \begin{tabular}{@{}l l c c@{}}
    \toprule
    Model & Hypothesis & LR statistic & p-value \\
    \midrule
    HBXII  & H\textsubscript{0}: $\alpha$ = 1 vs H\textsubscript{1}: H\textsubscript{0} is not true & 5.264 & 0.022 \\
    APBXII & H\textsubscript{0}: $\upsilon$ = c = 1 vs H\textsubscript{1}: H\textsubscript{0} is not true & 13.703 & 0.001 \\
    MOBXII & H\textsubscript{0}: $\upsilon$ = $\alpha$ = 1 vs H\textsubscript{1}: H\textsubscript{0} is not true & 8.807 & 0.012 \\
    BXII   & H\textsubscript{0}: $\upsilon$ = c = $\alpha$ = 1 vs H\textsubscript{1}: H\textsubscript{0} is not true & 18.384 & 3.67E-04 \\
    \bottomrule
    \end{tabular}
\end{table*}

\begin{table*}[ht!]
    \centering
     \footnotesize
    \caption{Likelihood ratio statistics for the cancer data}
    \label{Table:12}
    \begin{tabular}{@{}l l c c@{}}
    \toprule
    Model & Hypothesis & LR statistic & p-value \\
    \midrule
    HBXII  & H\textsubscript{0}: $\alpha$ = 1 vs H\textsubscript{1}: H\textsubscript{0} is not true & 2.070794 & 0.150143 \\
    APBXII & H\textsubscript{0}: $\upsilon$ = c = 1 vs H\textsubscript{1}: H\textsubscript{0} is not true & 35.87769 & 1.62E-08 \\
    MOBXII & H\textsubscript{0}: $\upsilon$ = $\alpha$ = 1 vs H\textsubscript{1}: H\textsubscript{0} is not true & 3.384205 & 0.184132 \\
    BXII   & H\textsubscript{0}: $\upsilon$ = c = $\alpha$ = 1 vs H\textsubscript{1}: H\textsubscript{0} is not true & 87.83823 & 1.01E-15 \\
    \bottomrule
    \end{tabular}
\end{table*}

\begin{table*}[ht!]
    \centering
     \footnotesize
    \caption{Likelihood ratio statistics for the device lifetime data}
    \label{Table:13}
    \begin{tabular}{@{}l l c c@{}}
    \toprule
    Model & Hypothesis & LR statistic & p-value \\
    \midrule
    HBXII  & H\textsubscript{0}: $\alpha$ = 1 vs H\textsubscript{1}: H\textsubscript{0} is not true & 13.595 & 2.27E-04 \\
    APBXII & H\textsubscript{0}: $\upsilon$ = c = 1 vs H\textsubscript{1}: H\textsubscript{0} is not true & 32.389 & 9.26E-08 \\
    MOBXII & H\textsubscript{0}: $\upsilon$ = c = 1 vs H\textsubscript{1}: H\textsubscript{0} is not true & 13.928 & 9.45E-04 \\
    BXII   & H\textsubscript{0}: $\upsilon$ = c = $\alpha$ = 1 vs H\textsubscript{1}: H\textsubscript{0} is not true & 52.980 & 1.85E-11 \\
    \bottomrule
    \end{tabular}
\end{table*}

\section{Conclusion}\label{Sec:7} 

In this study, we introduce and establish the Alpha Power Harris--G (\APHG) as a new family of continuous distributions. Furthermore, we use the Burr XII distribution as a baseline in the \APHG\ generator, and a new five-parameter \APHBXII\ distribution was derived and studied. The proposed distribution has several advantageous probabilistic characteristics and properties that make it a strong candidate for modelling lifetime data with monotone and non-monotone failure rates. Subsequently, a Monte Carlo simulation study was conducted, and the MLE method was used to estimate the unknown parameters of the model. Building on this, three real lifetime datasets were analyzed, and the goodness-of-fit metrics of the \APHBXII\ distribution were compared with those of four other competitive models. Although, it cannot be guaranteed that the proposed distribution will always provide the best fit, but the model may outperform other established ageing distributions in many practical scenarios. It is anticipated that the new \APHBXII\ distribution will gain wider utilization in reliability engineering and survival-time modelling. Lastly, the newly proposed \APHG\ generator may also serve as a useful tool for researchers interested in developing novel probability distributions.

\bibliographystyle{elsarticle-num}
\biboptions{sort&compress}

\bibliography{references}

\end{document}